%% file: main.tex
\documentclass[letterpaper,11pt]{article}
\usepackage{fullpage}
\usepackage[hmargin=1in,vmargin=1in]{geometry}
\usepackage[english]{babel}
\usepackage[utf8]{inputenc}
\usepackage[T1]{fontenc}
\usepackage{amsmath,amsthm}
\usepackage{mathrsfs}
\usepackage{algorithmic}
\usepackage{algorithm}
\usepackage{enumerate}  
\usepackage{subfig}
\usepackage{amssymb}
\usepackage{lipsum}
\usepackage{mathtools}
\usepackage{graphicx,color}
\usepackage{authblk}
\usepackage{hyperref}
\usepackage{booktabs} 
\usepackage{tikz,subfig}
\usetikzlibrary{matrix,arrows,decorations.pathreplacing,automata,positioning,decorations.pathmorphing}
\usepackage{float}

\newtheorem{lemma}{Lemma}
\newtheorem{theorem}{Theorem}
\newtheorem{definition}{Definition}

\newtheorem{claim}{Claim}

\title{A Distributed Laplacian Solver and its Applications to Electrical Flow and Random Spanning Tree Computation}

\author[1]{Iqra Altaf Gillani}
\author[1]{Amitabha Bagchi}
\affil[1]{Department of Computer Science and Engineering, IIT Delhi\\ \{iqraaltaf,bagchi\}@cse.iitd.ac.in}

\date{}

\usepackage{bm}
\usepackage{mdframed}
\usepackage{tikz} \usetikzlibrary{shapes,arrows}
\usepackage{algorithm,algorithmic}
\newenvironment{varalgorithm}[1]   {\algorithm\renewcommand{\thealgorithm}{#1}}   {\endalgorithm}
  
\newcommand{\ex}[1]{{\mbox{E}\left[#1\right]}}
\newcommand{\ind}[1]{{\mbox{\bf 1}_{\{#1\}}}}
\newcommand{\probm}[1]{\mathcal{P}[ #1 ]}
\newcommand{\pot}[1]{\bm{x}_{#1}}
\newcommand{\hpot}[1]{\hat{\bm{x}}_{#1}}
\newcommand{\prob}[1]{P\left[ #1 \right]}
\newcommand{\reff}[1]{R_{#1}}

\newcommand{\vol}[1]{d_{#1}}
\newcommand{\sink}{u_{\mbox{\scriptsize s}}}

\newcommand{\hemax}{\bm{\hat{\eta}}_{\mbox{\scriptsize max}}}
\newcommand{\heta}[1]{\bm{\hat{\eta}}_{#1}}
\newcommand{\thit}{t_{\mbox{\scriptsize hit}}}

\def\diam{\mbox{\sf diam}(G)}
\def\diams{\mbox{\sf diam}}

\def\RR{\mathbb{R}}
\def\nn0{\mathbb{N} \cup \{0\}}
\def\nonsinks{V \setminus \{\sink\}}

\providecommand{\keywords}[1]{\textbf{\textit{Keywords:}}#1}

\date{}
\begin{document}

\maketitle

\begin{abstract}
We use queueing networks to present a new approach to solving Laplacian systems. This marks a significant departure from the existing techniques, mostly based on graph-theoretic constructions and sampling. Our distributed solver works for a large and important class of Laplacian systems that we call ``one-sink'' Laplacian systems. Specifically, our solver can produce solutions for systems of the form $L\bm{x} = \bm{b}$ where exactly one of the coordinates of $\bm{b}$ is negative. Our solver is a distributed algorithm that takes $\widetilde{O}(\thit d_{\max})$ time (where $\widetilde{O}$ hides $\text{poly}\log n$ factors) to produce an approximate solution where $\thit$ is the worst-case hitting time of the random walk on the graph, which is $\Theta(n)$ for a large set of important graphs, and $d_{\max}$ is the generalized maximum degree of the graph. The class of one-sink Laplacians includes the important voltage computation problem and allows us to compute the effective resistance between nodes in a distributed setting. As a result, our Laplacian solver can be used to adapt the approach by Kelner and M\k{a}dry (2009) to give the first distributed algorithm to compute approximate random spanning trees efficiently.

\end{abstract}

\section{Introduction}
\label{sec:intro}
A number of fields including network analysis, computer vision, operations research, electrical engineering, machine learning, and computational biology, present important algorithmic problems that can be approached by solving a Laplacian system of equations. Since Spielman and Teng's~\cite{Spielman-STOC:2004} pioneering work proposing a quasi-linear time algorithm for solving Laplacian systems, Laplacian solvers have been able to improve longstanding bounds for a wide range of fundamental graph-theoretic problems, c.f., surveys by Spielman~\cite{Spielman-SICOMP:2011} and Vishnoi~\cite{Vishnoi-F&TTCS:2013}. Moreover, if we view the graph as an electrical network with resistances on edges, the effective resistance between two nodes can be easily derived by solving a Laplacian system. As a result, a series of exciting graph algorithmic developments that use electrical flow computation as a primitive also rely on Laplacian solvers, e.g., max flows \cite{Christiano-STOC:2011}\cite{Madry-FOCS:2016}, random spanning trees \cite{KelnerII-FOCS:2009}\cite{Madry-SODA:2015}, graph sparsification \cite{Spielman-SICOMP:2011}, and computation of current flow closeness centrality \cite{Li-WWW:2019}. 

Although the Laplacian systems are defined over a graph, and have been used to address problems that can arise naturally in distributed settings, most quasi-linear solvers are centralized. These solvers (e.g., Spielman and Teng~\cite{Spielman-STOC:2004}, Koutis et al. \cite{Koutis-FOCS:2011}, Cohen et al. \cite{Cohen-STOC:2014}) are based on techniques like Richardson iteration, Chebyshev polynomial-based method or conjugate gradient, combining them further with ideas like recursive preconditioning, spectral graph sparsifiers, low-stretch spanning trees and expanders for improved results. These techniques rely on use of matrix multiplication or matrix-vector multiplication as primitives, which presents a challenge when we think of implementing them in distributed settings on machines with limited computational resources. Departing from these graph-based constructions, Kyng and Sachdeva \cite{Kyng-FOCS:2016} adopted a different approach based on sparse Gaussian elimination using random sampling to give a nearly linear time solver.  However, this proposal is also centralized. 

Many such centralized solvers have been parallelized now like the work of \cite{Koutis-SODA:2007}, \cite{Blelloch-TCS:2014}, \cite{Peng-STOC:2014}. Lately, distributed algorithms for this problem have begun to appear in the literature,  e.g., the algorithms described by Tutunov et al. in \cite{Tutunov-arXiv:2015} and \cite{Tutunov-CDC:2017}. But these, like the algorithm of Zouzias and Freris \cite{Zouzias-ECC:2015} are also based on distributed implementations of existing techniques used for centralized solvers like the Kaczmarz method, Chebyshev polynomials. 
In general, due to the complex underlying machinery involved, all these algorithms are difficult to implement for practitioners, suggesting the need for a simple distributed solver. In this paper, we address that need by describing a simple distributed algorithm that uses no complicated primitives but can
solve a large and useful class of Laplacian systems that includes electrical flow. 

In particular, our algorithm represents {\em a completely new approach to Laplacian systems that sits at the intersection of queueing network theory, multidimensional Markov chain analysis, and random walk theory.} While random walks have been used in the literature for solving Laplacian systems~\cite{Chung-IM:2015,Andoni-ITCS:2019}, their use is restricted to estimating matrix powers via sub-sampling, and these approaches focus on returning only one or a few coordinates of the solution vector. 

An important feature of our distributed Laplacian solver is that it does not require any complex graph theoretic constructions or matrix-vector operations at the nodes. It is based on simple counting, making it suitable for networks with low-end nodes with very limited computation, e.g., sensor networks or IoT networks.

\begin{table*}[!t]
\caption[]{Time complexity of various distributed Laplacian solvers}
\label{table:solver_time}
\begin{center}
\scalebox{0.7}{
\begin{tabular}{lccc}
\hline
\\
\textbf{Graph}  
& \textbf{DRW-LSolve}  &\textbf{Zouzias et al. \cite{Zouzias-ECC:2015}} &\textbf{Tutunov et al. \cite{Tutunov-CDC:2017}}\\
\hline\\
Complete $(K_n)$ &\bm{$\widetilde{O}(n^2)$} &\bm{$\widetilde{O}(n^2)$} &$\widetilde{O}(n^{5/2})$\\
Star &\bm{$\widetilde{O}(n^2)$} &\bm{$\widetilde{O}(n^2)$} &$\widetilde{O}(n^{5/2})$\\
Binary tree &\bm{$\widetilde{O}(n)$} &$\widetilde{O}(n^2)$ &$\widetilde{O}(n^{3/2})$\\
$\sqrt{n}\times\sqrt{n}$ Grid  &\bm{$\widetilde{O}(n)$} &$\widetilde{O}(n^2)$ &$\widetilde{O}(n^{3/2})$\\
Erdos-Renyi &\bm{$\tilde{O}(n)$} &\bm{$\tilde{O}(n)$} &$\tilde{O}(n^3/2)$\\
(with $p \geq c \log(n)/n$ for constant $c>1$) & &\\ 
RGG-Disc &\bm{$\tilde{O}(n)$} &$\tilde{O}(n^2)$ &$\tilde{O}(n^3/2)$\\
(with $r= c\sqrt{\log n/n}$ for constant $c>1$) & &\\
RGG-cNN &\bm{$\tilde{O}(n)$} &$\tilde{O}(n^2)$ &$\tilde{O}(n^3/2)$\\
(with constant $c > 11$) & &\\
Ramanujan  $(d,n)$ &$\widetilde{O}(dn)$ &\bm{$\widetilde{O}\left(\frac{dn}{d-2\sqrt{d-1}}\right)$} &$\widetilde{O}(dn^{3/2})$\\
Path $(P_n)$ &$\widetilde{O}(n^2)$ &$\widetilde{O}(n^3)$ &\bm{$\widetilde{O}(n^{3/2})$}\\
\\
\hline
\end{tabular}
}
\end{center}
\end{table*}

As we can see in Table~\ref{table:solver_time} our solver outperforms state of the art distributed methods on most common topologies. The number of rounds taken by our algorithm is proportional to $\thit$, the worst-case hitting time of the natural random walk on the graph. While this quantity can be quite large for some topologies, it is $\Theta(n)$ for several important classes of graphs like the expanders, random geometric graphs which model real-world setting, the complete graph $(K_n)$, the star graph, the hypercube, the $\sqrt
{n}\times\sqrt{n}$ grid, the Erd\H{o}s-R\'{e}nyi graphs with edge probability $p\geq \frac{\log n}{n}$.

\subsection{Our Approach}
\label{subsec:our_approach}
Specifically, we formulate a stochastic problem that captures the properties of Laplacian systems of the form $L\bm{x} = \bm{b}$ with a constraint that only one element in $\bm{b}$ is negative. We call such systems ``one-sink'' Laplacian systems since the stochastic process can be viewed as a network in which some nodes generate data packets, and there is a single sink that collects all the packets that it receives. The network transmission model used for such a process follows gossip strategy \cite{Boyd-TON:2006}, whereby, at every step, each node with a non-empty queue can transmit a data packet only to a single neighbor. This choice of selecting neighbors is probabilistic and depends on the transition probabilities of the reversible random walk. The class of one-sink Laplacians contains the important electrical flow problem. We call our stochastic process the ``Data Collection problem'' and show that this process has an equivalence to the one-sink Laplacian system at stationarity, provided the associated multidimensional Markov chain is ergodic and has a stationary distribution. 

The technical challenges that arise in deriving an algorithm from this observation are: (1) We have to ensure that our Data Collection problem is ergodic, (2) that we can get close to the stationarity in reasonable time and (3) once we are close to the stationarity we can estimate the solution in good time. We show that in its stationary state, the Data Collection system naturally contains a solution to the Laplacian system. In particular, the vector whose $v$th coordinate is proportional to the probability at stationarity of the queue at $v$ being non-empty is a solution to $L\bm{x} = \bm{b}$. Thus, we can easily estimate this solution by first allowing the Markov chain to be close to the stationarity and then counting the proportion of time slots for which the queues are empty. Our approach can be visualized as a probabilistic version of iterative algorithms like the Richardson's iteration wherein each iteration represents a communication round in which each node makes a probabilistic decision of choosing a neighbor. After a specified number of such iterations required: (i) to be close to the stationary distribution of Markov chain, and (ii) computing the queue occupancy probability of nodes, the approximate solution to the Laplacian system is obtained. \label{comm:insight}


\subsection{Our Model}
\label{subsec:our_model}
Our setting is as follows: we are given an undirected connected weighted graph $G=(V,E,w)$ with $|V|=n$ nodes each having a unique identifier of size $O(\log n)$, $|E|=m$ edges, positive weight function $w:E\rightarrow \mathbb{R}_+$, adjacency matrix $A$ such that $A_{uv}=w_{uv}$, if $(u,v)\in E$ and $0$ otherwise. Let $D\in \mathbb{R}^{|V|\times |V|}$ be the diagonal matrix of generalized degrees such that $D_{uu}=d_u$, where $d_u=\sum_{v:(u,v)\in E}w_{uv}$. Then, the corresponding Laplacian matrix is $L=D-A$. The natural random walk defined on $G$ has transition matrix $\mathcal{P}=D^{-1}A$ where $\probm{u,v} = w_{uv}/d_u$ for $(u,v)\in E$ and $0$ otherwise. If $(X_i)_{i \geq 0}$ is an instance of the random walk on $G$ then the {\em worst-case hitting time}, $\thit$, of this random walk is defined by $\max_{u,v\in V} \ex{t : X_0 = u, X_t = v, X_i \ne v \mbox{ for } 0 < i < t}$, i.e., the expected time taken for the random walk started at $u$ to reach $v$ maximized over all pairs of vertices.

We assume a synchronous model of communication, which we call as \emph{GP-CONGEST} based on its relation with the standard \emph{CONGEST} model, and \emph{gossip} algorithms. In this model, all nodes wake up simultaneously and the time is divided into rounds. In each round, all nodes in the network can send a message of size $O(\log n)$ to one of its neighbors and perform local computation based on the information obtained from the messages received from its neighbors in previous rounds. However, note that by definition of this model, each node can receive at most $d_{\max}$ messages in one round, one from each of its neighbors, where $d_{\max}$ is the generalized maximum degree of the graph. So, each round will require in worst-case $d_{\max}$ time to simulate. Thus, any algorithm's time complexity under this model is the maximum degree times the number of synchronous rounds required until the given task is complete.
We can formally define the metric as follows.

\begin{definition}[Time complexity under GP-CONGEST model]
\label{def:comp_metric}
The time complexity of any algorithm under the GP-CONGEST model is defined as the time taken by the synchronous communication rounds to complete the required task. Wherein each round takes in worst-case $d_{\max}$ time, and during which all nodes in parallel can
(i) send a $O(\log n)$ sized message to one of their neighbors,
(ii) receive $O(\log n)$ sized messages from their neighbors, and
(iii) perform bounded local computations.
\end{definition}

Our model is similar to the standard CONGEST model \cite{Peleg-BOOK:2000*} as we limit the size of messages to $O(\log n)$. However, it is more restrictive as in any round, instead of allowing message exchange with all the neighbors as done in the CONGEST model, our model allows nodes to send message to only one neighbor, similar to gossip algorithms. This restriction of allowing nodes to send only one data packet helps in easier theoretical analysis of our setting. However, for practical purposes, if we remove this restriction, this results in speed up in our solver's running time.

\subsection{Our Results}
\label{subsec:our_results}
The main result of our paper is as follows.
\begin{theorem}[Distributed Laplacian Solver]
\label{thm:rand_walk_solver}
Given an undirected connected graph $G = (V,E,w)$ and $\bm{b} \in \mathbb{R}^{|V|}$ such that for $1 \leq i \leq n-1, \bm{b}_i \geq 0$, and $\bm{b}_n = -\sum_{i=1}^{n-1} \bm{b}_i$, the corresponding Laplacian system is
\begin{equation}
\label{eq:solver-laplacian-system}
\bm{x}^TL = \bm{b}^T.
\end{equation}
For a user-defined granularity factor $\kappa \in (0,1)$ and polynomially bounded error parameter $\epsilon\in (0,1)$ such that $\epsilon<1/2$, there is a distributed solver  \ref{alg:rw_solver}$(w,\bm{b},\epsilon,\kappa)$ working in the GP-CONGEST model that takes 
$O\left(d_{\max}\left(\thit\log \epsilon^{-1}+ \frac{\log n}{\kappa^2\epsilon^2}\right)\log\frac{d_{\max}}{\lambda_2^L}\right)$
time, to produce a vector $\bm{\hat{x}}$ that is an approximate solution to Eq.~\eqref{eq:solver-laplacian-system}, where $d_{\max} = \max_{u\in V} d_u$, $\lambda_2^L$ is the second smallest eigenvalue of the Laplacian matrix $L$, and $\thit$ is the worst-case hitting time of the natural random walk on $G$.

Further, $\bm{\hat{x}}$ has the properties that
(i) $\hpot{v}$ is the estimate at node $v \in V$,
(ii) $\hpot{i} > 0$, $1 \leq i \leq n-1$, $\hpot{n} = 0$, and,
(iii) if $\bm{x}$ is an exact solution to Eq.~\eqref{eq:solver-laplacian-system} such that $\pot{i} > 0$, $1 \leq i \leq n-1$, $\pot{n} = 0$, 
then, $|\hpot{i} - \pot{i}| \leq \epsilon\pot{i},$
whenever $\kappa(1+\epsilon)< \pot{i} d_i$.
\end{theorem}
Our proposed solver returns an approximate solution to the Laplacian system with $\epsilon$-relative error guarantee only for those node estimates that are greater than $\kappa(1+\epsilon)/d_i$ where $\kappa$ is the user-defined granularity factor. Moreover, the running time of our solver has a logarithmic dependence on the second smallest eigenvalue of Laplacian matrix $\lambda_2^L$, also known as \emph{Fiedler value or algebraic connectivity} of graph and a direct dependence on the maximum degree $d_{\max}$ and the worst-case hitting time $\thit$. However, in general we have $\thit \leq \frac{d_{max}}{d_{min}}\frac{m}{\lambda_2^L}$ \cite{Aldous-BOOK:2002}, so the higher the value of algebraic connectivity,  the better is the running time of our solver.

In Table~\ref{table:solver_time}, we compare the running time of our solver with the existing distributed algorithms of Zouzias and Freris \cite{Zouzias-ECC:2015} and Tutunov et al. \cite{Tutunov-arXiv:2015}.
It is clear from the table that our algorithm, in general, outperforms the existing best known distributed solvers as it gives the best running time for most of the cases, except for graphs with bottleneck like the path graph $(P_n)$. 

The algorithm mentioned in Theorem~\ref{thm:rand_walk_solver} can be directly applied to compute effective resistance $\reff{u,v}$ between a pair of node $u,v$ by considering $u$ as the source node and $v$ as the sink. 
\begin{theorem}[Effective Resistance Computation]
\label{thm:eff_res_compt}
Given an undirected connected graph $G=(V,E, w)$ with a positive weight function and $|V|=n$ nodes there is a distributed algorithm working in the GP-CONGEST model to compute the effective resistance $\reff{u,v}$ between any pair of nodes $u ,v$ within a polynomially bounded $\epsilon$ error such that $\epsilon<1/2$ in 
$O\left(d_{\max}\left(\thit(u,v)\log \epsilon^{-1}+ \frac{d_{\max}^2\log n}{d_{\min}^2\epsilon^2}\right)\log\frac{d_{\max}}{\lambda_2^L}\right)$
time where $\thit(u,v)$ is the expected time taken by the random walk starting from node $u$ to hit node $v$, $d_{\max}$, $d_{\min}$ are the maximum and minimum generalized degrees of graph $G$ respectively, and $\lambda_2^L$ is the second smallest eigenvalue of the Laplacian matrix of graph.
\end{theorem}
This distributed way of finding effective resistance in turn helps to compute the electrical flows. Using such distributed electrical flow computation as a subroutine we describe an algorithm for distributed random spanning tree generation by adapting the centralized algorithm of Kelner and M\k{a}dry~\cite{KelnerII-FOCS:2009}.
\begin{theorem}[Random Spanning Tree Generation]
\label{thm:appx_rst_gen}
Given an undirected connected bounded-degree graph $G=(V,E)$ with $|V|=n$ and $|E|=m$ edges, our Distributed RST Generation Algorithm computes random spanning tree of $G$ from a distribution that is $(1+\epsilon)$ of the uniform distribution with a running time of $\widetilde{O}\left(\frac{m}{\sqrt{n}}\thit\log \epsilon^{-1}\right)$
under the GP-CONGEST model where $\thit$ is the worst-case hitting time of random walk on graph $G$.
\end{theorem}
To the best of our knowledge, this is the first distributed algorithm for this problem. Moreover for a large class of graphs with $\thit=\Theta(n)$, compared to the best known centralized algorithms which work best for either sparse graphs like the one proposed in \cite{Madry-SODA:2015}, or dense graphs like the one in \cite{Durfee-STOC:2017}, our algorithm gives reasonable bounds irrespective of the graph sparsity. 

Note that, although the main results in this paper are for the one-sink Laplacian systems. We also show a direct extension of our proposed algorithm to general Laplacian systems in Section~\ref{sec:general_solver}.

\subsection{Paper Organization}  
We briefly review the literature in Section~\ref{sec:rel_work}. Then, we introduce the notion of ``One-sink" Laplacian systems in Section~\ref{subsec:elec_netw_flows} followed by a discussion about its equivalence with a stochastic process on a graph, the Data Collection process,  in Section~\ref{subsec:data_coll_gossip}. In Section~\ref{sec:distributed_solver}, we present a distributed solver to use this equivalence to solve one-sink Laplacian systems. We also show an extension of our proposed solver to general Laplacian system in Section~\ref{sec:general_solver}. In Section~\ref{sec:rst_gen}, we present an application of our solver for $\epsilon$-random spanning tree generation. 
Finally, we conclude in Section~\ref{sec:conc} with some general remarks and some directions for future work.

\section{Related Work}
\label{sec:rel_work}
After the breakthrough paper by Spielman and Teng \cite{Spielman-STOC:2004} wherein the Laplacian system of equations are approximately solved in $\widetilde{O}(m \log^c n \log \epsilon^{-1})$ time where $c$ is a constant and $\epsilon$ denotes the error, an extensive literature has been developed in which a number of quasi-linear time solvers like the ones in \cite{Koutis-FOCS:2011}, \cite{Kelner-STOC:2013}, \cite{Cohen-STOC:2014} have been proposed, each improving the value of exponent $c$ in the running time (see survey by \cite{Spielman-SICOMP:2011}, \cite{Vishnoi-F&TTCS:2013} for details). As discussed before, most of these algorithms rely on similar techniques using Chebyshev iteration or conjugate gradient with recursive preconditioning, low-stretch spanning trees, expanders, and spectral graph sparsifiers.
More recently, Kyng and Sachdeva \cite{Kyng-FOCS:2016} proposed a different kind of nearly linear time solver that relies on random sampling by performing sparse approximate Gaussian elimination. Most of these solvers were proposed as a theoretical work with no working implementation. However, by now a number of experimental papers like the work of Boman et al. \cite{Boman-INFO:2016}, Hoske et al. \cite{Hoske-SEA:2015}, comparative study by Boman, Deweese and Gilbert \cite{Boman-ALENEX:2016} discuss these implementations, which are now also available as part of \emph{Laplacians.jl} package written by Spielman \cite{Spielman-URL:2017}. In addition to these theoretical solvers, a number of robust practical solvers based on multigrid technique have been proposed like Combinatorial Multigrid (CMG) by Koutis et al. \cite{Koutis-CVIU:2011}, Lean Algebraic Multigrid (LAMG) by Livne and Brandt \cite{Livne-SICOMP:2012}, and Degree-aware Rooted Aggregation (DRA) by Napov and Notay \cite{Napov-SICOMP:2017} which offer fast empirical running times. All these methods are however centralized.

Motivated by problems involving massive graphs, there has been some work in parallelizing these solvers. Koutis and Miller \cite{Koutis-SODA:2007} proposed a parallel solver for planar graphs with nearly linear work and $m^{\frac{1}{6}}$ depth. Later, Blelloch et al.~\cite{Blelloch-TCS:2014} extended this result to general graphs while improving the depth to $m^{\frac{1}{3}}$. This was further improved by Peng and Spielman \cite{Peng-STOC:2014}, who gave an efficient algorithm requiring nearly linear work and polylogarithmic depth.
However, because of the shared memory model, these algorithms cannot be easily distributed. Similarly, a method that is amenable to being parallelized is that of Becchetti et al.~\cite{Becchetti-AAMAS:2018}, who use a random walk based token diffusion process similar in spirit to our method to present a solver for the specific case of the electrical flow Laplacian. Their method
cannot be adapted easily to the distributed setting since the number of tokens entering and leaving a vertex are potentially unbounded. Chung and Simpson \cite{Chung-IM:2015} also use random walks to give a local linear solver. However, their use is only restricted to computing powers of transition matrix. Also, their algorithm works only for Laplacian systems satisfying certain boundary conditions and returns solution only for a subset of coordinates. The idea of using random walks for power estimation of a matrix to approximate matrix inverse has also been used by Andoni et al. \cite{Andoni-ITCS:2019}. In particular, they also use random walks only for estimating powers of normalized adjacency matrix and focus only on one or a few coordinates of solution vector. Recently, Konolige and Brown \cite{Konolige-PASC:2018} also proposed a multigrid based parallel solver in a distributed memory setting which shows improved empirical results compared to the existing solvers, however with no theoretical guarantees on the worst-case running time. 

Working towards a completely distributed algorithm, Tutunov et al. \cite{Tutunov-arXiv:2015} propose a distributed implementation of Peng and Spielman's \cite{Peng-STOC:2014} solver with a time complexity of $O\left(d_{max}n^3\frac{w_{max}}{w_{min}}\log \epsilon^{-1}\right)$ where $d_{max}$ is the maximum degree, $w_{max}$, $w_{min}$ are the maximum and minimum edge weights respectively. They further improve this result to $O\left(d_{max}n^{3/2}\frac{w_{max}}{w_{min}}\log \epsilon^{-1}\right)$ \cite{Tutunov-CDC:2017} by using Chebyshev polynomials. In both the algorithms, Tutunov et al. use the bound of $O\left(n^3\frac{w_{max}}{w_{min}}\right)$ on the condition number as its computation in distributed fashion is difficult and requires global information (see \cite{Tutunov-THESIS:2017} for details). However, unlike the gossip setting \cite{Boyd-TON:2006} wherein each node can exchange messages with only one neighbor, both algorithms allow nodes to exchange messages with all of their neighbors. Further, Zouzias and Freris \cite{Zouzias-ECC:2015} adapted the Kaczmarz iteration method for solving Laplacians~\cite[Chap. 14]{Vishnoi-F&TTCS:2013} to a gossip scenario similar in working to our GP-CONGEST model. Their approach is based on using the gossip model as a means of achieving consensus for solving Laplacians as a least-square estimation problem. As expected, their algorithm converges to an approximate solution in $O\left(\frac{m}{\lambda_2^L}\right)$ iterations with each iteration involving an inner product computation which takes at most $O(d_{max})$ time. Our underlying method, which also uses a similar gossip setting, achieves a running time that in addition to $d_{max}$ factor, depends on the worst-case hitting time of the natural random walk defined on the graph.
More recently, Rebeschini and Tatikonda~\cite{Rebeschini-JMLR:2019*} analyzed the performance of the distributed message-passing algorithm ``min-sum'' to solve the electrical flow problem, which is a proper subset of the class of Laplacian systems we consider, but, their findings are largely negative, i.e., their method doesn't converge in most classes of graphs.

The literature on electrical flow and its applications is too vast to survey here so we just mention that the Laplacian representation of the electrical flow problem has evolved it as a popular subroutine in solving various graph-related problems like maximum flow computation~\cite{Christiano-STOC:2011}, graph sparsification~\cite{Spielman-SICOMP:2011}, random spanning tree generation \cite{KelnerII-FOCS:2009}. Our electrical flow computation method can be used to give a distributed version of the algorithm of Kelner and M\k{a}dry~\cite{KelnerII-FOCS:2009} which is the first distributed algorithm for the random spanning tree problem to the best of our knowledge, all prior work being centralized, e.g., the algorithms proposed in \cite{Madry-SODA:2015,Durfee-STOC:2017,Schild-STOC:2018}.

\section{``One-sink'' Laplacian Systems and an Equivalent Stochastic Process}

\label{sec:elec_flow_lap}
In this section, we first introduce the notion of ``One-Sink'' Laplacian systems. After that, we discuss a stochastic process: Data Collection on a graph and show how the steady-state equations of this setting exactly mimic a Laplacian system.

\subsection{``One-Sink'' Laplacian Systems}
\label{subsec:elec_netw_flows}


The Laplacian $L$ is a singular matrix with its null space being the subspace generated by the vector $\bm{1}$ which contains 1 in all its coordinates. Hence for any vector $\bm{x}$,\footnote{In the following we use bold letters, e.g., $\bm{x}$ for column vectors and denote row vectors as the transpose of column vectors, e.g., $\bm{x}^T$.} the inner product  $\langle L\bm{x}, \bm{1}\rangle =0$. Therefore, when we consider equations of the form $L\bm{x} = \bm{b}$, we must have that $\langle \bm{b}, \bm{1}\rangle = 0$. Within this space we focus on the special case in which there is exactly one coordinate $i$ such that $\bm{b}_i < 0$. In such a case, clearly $\bm{b}_i = - \sum_{j \ne i} \bm{b}_j$. For reasons that will be clear shortly we will call such vectors {\em one-sink} vectors and we will refer to the Laplacian systems of the form $L\bm{x}=\bm{b}$
as {\em one-sink} Laplacian systems. 
Noting that since $L$ is symmetric we can rewrite $\bm{x}^TL=\bm{b}^T$.
This latter form will be useful to us since in Section~\ref{subsec:data_coll_gossip} where we will define a stochastic process on a network whose stationary state will give us a solution in this form. At that point the ``one-sink'' terminology will also be clarified. 

\subsubsection{One-sink Laplacians and effective resistance} The importance of one-sink Laplacians can be estimated by observing that the very widely used notion of the effective resistance falls within it. To see this, recall (or see, e.g., \cite{Vishnoi-F&TTCS:2013}) that if $L^+$ is the Penrose-Moore pseudo-inverse of $L$ then for an edge $u,v \in V$, then effective resistance between $u$ and $v$, 
\begin{equation}
\label{eq:reff_lap}
R_{u,v} = (\bm{e_u} - \bm{e}_v)^T L^+ (\bm{e}_u - \bm{e}_v),
\end{equation}
where $\bm{e}_w$ is the vector with 1 at coordinate $w$ and 0 elsewhere. Set $\bm{b} =  \bm{e}_u - \bm{e}_v$, which is clearly a one-sink vector, and post-multiply both sides of $\bm{x}^TL=\bm{b}^T$ by $L^+$
to obtain
$\bm{x}^T=(\bm{e}_u - \bm{e}_v)^TL^+,$
and observe that the RHS of Eq.~\eqref{eq:reff_lap} is now simply $\bm{x}_u - \bm{x}_v$. So the effective resistance can be easily obtained by solving a one-sink Laplacian system. 

\subsection{An Equivalent Stochastic Process: Data Collection }
\label{subsec:data_coll_gossip}
We define a stochastic process on a graph that has the property that at stationarity its steady state behavior provides the solution of a one-sink Laplacian system. We now explain this stochastic process and the equivalence, since our distributed algorithm essentially involves ensuring this system has reached close to the stationarity and then estimating the steady state solution. 

Formally the stochastic process can be defined as follows.
\begin{definition}[Data Collection process]
\label{def:data_colln_process}
On a connected undirected weighted network $G=(V,E,w)$ we identify a distinguished sink node, $\sink$ that passively collects data, and a set of data sources $V_s\subseteq \nonsinks$. Each node in $\nonsinks$ has a queue in which it can store data packets. We define a relative rate vector $\bm{J} \in \mathbb{R}^{|\nonsinks|}$:  for a source $v \in V_s$, $\bm{J}_v > 0$, for all other nodes in $\nonsinks$, $\bm{J}_v = 0$, and  $\bm{J}_{\sink} = -\sum_{v\in \nonsinks} \bm{J}_v$. 

Now, given a parameter $\beta > 0$ such that $\beta \bm{J}_v \leq 1$ for all $v \in V_s$, we define a discrete-time multidimensional Markov chain $\{Q_t^\beta\}_{t\geq 0}$ supported on $\left(\nn0\right)^{|\nonsinks|}$ that we call as the {\em Data Collection process} with parameter $\beta$. On this network:
\begin{itemize}
\item For $t \geq 0, v \in \nonsinks$, $Q^\beta_t(v)$ denotes the size of the queue at node $v$ at time $t$.
\item At each $t \geq 0$, node $v \in V_s$ generates a new packet according to independent Bernoulli distribution with parameter $\beta \bm{J}_v$ and places it in its queue.
\item At each $t \geq 0$ each node $u$ in $\nonsinks$ picks one packet uniformly at random from its queue if its queue is non empty. It picks a neighbor $v$ according to probability $w_{uv}/d_u$ and transmits that packet to $v$. 
\item On transmission, the packet is removed from $u$'s queue. If $v \in \nonsinks$ the packet is placed in $v$'s queue. If $v$ is $\sink$ then the packet is sunk,  i.e., it is  removed from the system.
\end{itemize}
\end{definition}
Clearly if the rate controlling parameter $\beta$ is too high, $ ||Q_t^\beta||$ will tend to infinity in the limit, so for this definition to be useful there needs to be a regime of $\beta$ values wherein this process achieves its steady state. In \cite{GillaniII-arXiv:2019} we show that such a regime does exist.
\begin{theorem}[Theorem~1, Gillani et al. \cite{GillaniII-arXiv:2019})]
\label{thm:gillani-stability}
For the Data Collection process there exists a $\beta^*>0$ such that the multidimensional Markov chain $\{Q_t^\beta\}_{t\geq 0}$is ergodic for all $\beta < \beta^*$ and for  $\beta \geq \beta^*$ the chain is non-ergodic. Specifically, for $\beta < \beta^*$,
\begin{itemize}
\item  $\{Q^{\beta}_t\}_{t\geq 0}$ has a stationary distribution and 
\item  $\lim_{t \rightarrow \infty} |Q_t^\beta(v)|$ is a finite constant depending on $\beta$ for each $v \in \nonsinks$.
\end{itemize}
\end{theorem}
We prove this theorem using the induction-based technique developed by Georgiadis and Szpankowski \cite{Georgiadis-QS:1992} that intrinsically uses coupling of stochastic processes. 
We will drop the superscript $\beta$ from the queue notation where ever the context is understood.

\subsubsection{Establishing the equivalence to one-sink Laplacian systems} 
\label{subsubsec:equivalence}
We now show the equivalence of the evolution of this stochastic process with the one-sink Laplacian system of $\bm{x}^TL=\bm{b}^T$. 
We can write the basic one step queue evolution equation for any node $u\in V$ under the Data Collection process as
\begin{align}
\ex{Q_{t+1} (u) \mid Q_t (u)} &= Q_t (u) - \ind{Q_t (u) > 0} \sum_{v:v \sim u} \probm{u,v} +\sum_{v:v \sim u} \probm{v,u} \ind{Q_t (v) > 0}+A_t(u) 
\label{eq:one_step_queue_eq}
\end{align}
where the second and third term on the right-hand side of the above equation represents the transmissions sent to and received from the neighbors respectively and $A_t(u)$ is the number of packets generated at $u$, which is
independent of the queue size of node $u\in V$ at any time. In particular, $A_t(u)$ is 0 if $u \notin V_s$ and is 1 with probability $\beta \bm{J}_v$ if $v \in V_s$.
Now, taking expectations on both sides of Eq.~\eqref{eq:one_step_queue_eq} and let $\eta^t_u = \prob{Q_t(u) > 0}$ be the queue occupancy probability of node $u$ and observing that $\ex{A_t(u)}=\beta \bm{J}_u$, we have
\begin{align}
\ex{Q_{t+1} (u)} &= \ex{Q_t (u)}  - \eta_u^t \sum_{v:v\sim u}\probm{u,v} + \sum_{v:v\sim u}\probm{v,u}\eta_v^t+\beta \bm{J}_u.
\label{eq:exp_one_step_queue_eq}
\end{align}
By Theorem~\ref{thm:gillani-stability}, we know that for an appropriately chosen value of $\beta$ the process has a steady state. In a steady state $\ex{Q_t(u)}$    is a constant, so if we let $\bm{\eta}_u$ be the queue occupancy probability of node $u$ at the stationarity, then we have the steady-state equation for the given node as
\begin{equation}
- \bm{\eta}_u \sum_{v:v\sim u}\probm{u,v}  + \sum_{v:v\sim u}\probm{v,u}\bm{\eta}_v+\beta \bm{J}_u= 0.
\label{eq:steady_state_queue_eq_gen}
\end{equation}
Rewriting this in vector form we get 
\begin{equation}
\bm{\eta}^T(I-\mathcal{P})= \beta \bm{J}^T.
\label{eq:vector_queue_eq_1}
\end{equation}
Since $\mathcal{P}=D^{-1}A$, this can be rewritten as
\begin{equation}
\bm{x}^TL=\beta \bm{J}^T
\label{eq:steady-state_pot}
\end{equation}
where $\bm{x}^T=\bm{\eta}^TD^{-1}$ is a row vector such that $\pot{u}=\bm{\eta}_u/\vol{u}$ for all $u$ where $\bm{\eta}_u$ is the steady-state queue occupancy probability. Comparing this with the Laplacian equation $\bm{x}^TL=\bm{b}^T$
we see that they are identical, and so $\bm{\eta}^TD^{-1}/\beta$ is a solution for $\bm{x}^TL=\bm{b}^T$.

However, since $\beta \bm{J}_v \leq 1$ on the right hand side of Eq.~\eqref{eq:vector_queue_eq_1}, to handle general $\pmb{b}\in \mathbb{R}^{|V|}$ an appropriate scaling of the obtained solution from the Data Collection process needs to be done.
In summary, we can use the procedure shown in Figure~\ref{fig:schematic_steps} for deriving a solution to $\bm{x}^TL=\bm{b}^T$ 
by solving the Data Collection problem $\bm{x}^TL=\beta\bm{J}^T$ (Eq.~\eqref{eq:steady-state_pot}). We mark with an asterisk (*) those steps which are non-trivial.

\begin{figure}[ht!]
\centering
\begin{mdframed}
\begin{enumerate}
\item Create the network $G = (V, E, w)$ corresponding to the Laplacian $L$.
\item Since $\pmb{b}_v < 0$ for a single $v$, make that $v$ the sink for the Data Collection process. Now given $v$ is the chosen sink, set $\bm{J}=\frac{1}{\sum_{i\neq v}\bm{b}_i}\bm{b}$.
\item (*) Find a value of $\beta$ for which the Data Collection process is ergodic.
\item (*) Let the Data Collection process converge to stationarity.
\item Measure the vector $\bm{\eta}$ by computing the fraction of time steps each queue is occupied.
\item Return the solution $\frac{\sum_{i\neq v}\pmb{b}_i}{\beta}\pmb{\eta}D^{-1}$.
\end{enumerate}
\end{mdframed}
\caption{Steps for computing solution to the one-sink Laplacian equation using Data Collection process.}
\label{fig:schematic_steps}
\end{figure}

Finding the appropriate $\beta$ (Step 3) is non-trivial. Reaching stationarity is impossible (Step 4), so we have to figure out how close we need to be. The measurements made in Step 5 will also introduce error. In Section~\ref{sec:distributed_solver}, we present an algorithm that shows how to deal with these challenges. 

\subsubsection{A rate lower bound for the Data Collection process}
We will see in Section~\ref{sec:distributed_solver} that executing the program given in Figure~\ref{fig:schematic_steps} will take time that depends inversely on the value of $\beta^*$. For proving upper bounds on the time we will need the following lower bound on $\beta^*$. 
\begin{lemma}
\label{lem:beta_bound}
If $\beta^*$ is the value such that the multidimensional Markov chain $\{Q_t^\beta\}_{t\geq 0}$ associated with a Data Collection process on $G=(V,E,w)$ with sink $\sink$, source set $V_s$, and relative rate vector $\bm{J}$ is ergodic for $\beta < \beta^*$ and non-ergodic for $\beta \geq \beta^*$then $\beta^*\geq \frac{\lambda_2^L}{2d_{\max}\sum_{v\neq \sink}\bm{J}_v}.$
\end{lemma}
\begin{proof}[Proof of Lemma~\ref{lem:beta_bound}]
For the given Laplacian steady-state equation $\bm{x}^TL=\beta \bm{J}^T$, let $\bm{x}$ be {\em any} solution and  $\bm{y}$ be a canonical solution, i.e., the solution for which $\bm{y}^T \bm{1} =0$, i.e.,  $\sum \bm{y}_i = 0$. So, there must be some $w\in \RR$ such that $\bm{x} = \bm{y} + w \bm{1}$. Now, let $\bm{\psi}_i^T, 1 \leq i \leq n$ be the normalized (left) eigenvectors of $L$ corresponding to the eigenvalues $0=\lambda_1^L \leq \lambda_2^L \leq \cdots \leq \lambda_n^L$ such that $\bm{\psi}_i^T \bm{\psi}_i = 1, 1 \leq i \leq n$. Moreover, note that $\bm{\psi}_1 = \bm{1}$, so, $\bm{\psi}_i^T \bm{1} = 0$ for $2 \leq i \leq n$. So, rewriting the steady-state equation in terms of these eigenvectors we have  
\[ \sum_{i=2}^n \lambda_i^L \left(\bm{y}^T \bm{\psi}_i\right) \bm{\psi}_i = \beta \bm{J}^T.\]
Taking norms on both sides we get 
\begin{equation}
\lambda_2^L \left\Vert\sum_{i=2}^n \left(\bm{y}^T \bm{\psi}_i\right) \bm{\psi}_i\right\Vert \leq\left \Vert \sum_{i=2}^n \lambda_i^L \left(\bm{y}^T \bm{\psi}_i\right) \bm{\psi}_i \right\Vert = \left\Vert \beta \bm{J}^T\right\Vert.
\label{eq:norm_eq}
\end{equation}

Now, since we know $\bm{y}^T \bm{1}$ is 0 and the eigenvectors $\bm{\psi}_i,  2 \leq i \leq n$ span the subspace orthogonal to $\bm{1}$, therefore, $ \left\Vert \sum_{i=2}^n \left(\bm{y}^T \bm{\psi}_i\right) \bm{\psi}_i \right\Vert = \Vert \bm{y}\Vert$.
Putting this back in \eqref{eq:norm_eq} we have
\begin{equation}
\label{eq:upper-lower-norm}
\lambda_2^L \Vert\bm{y}\Vert \leq \beta \Vert\bm{J}^T\Vert.    
\end{equation}

We now try to find a lower bound on $\Vert\bm{y}\Vert$ for the specific case where $\beta = \beta^*$, the maximum stable rate for the data collection process. Let $\bm{x}$ be the solution produced by the data collection process at $\beta = \beta^*$. Since $\bm{x}^T D$ is a vector of queue occupancy probabilities in the data collection scenario, we know that 
\begin{enumerate}
    \item $\bm{x}_i \geq 0, 1 \leq i \leq n$.
    \item $\min_{i=1}^n \bm{x}_i = 0$.
    \item $\max_{i=1}^n d_i \bm{x}_i = 1$.
\end{enumerate}

From Fact 2 of this list we can deduce that if $\bm{x} = \bm{y} + w \bm{1}$ then $ w = - \min_{i=1}^n \bm{y}_i$. Now, putting this into Fact 3 of the list we get $\max_{i=1}^n d_i \left(\bm{y}_i  - \min_{j=1}^n \bm{y}_j \right) = 1$.
So, we get that 
\[\max_{i=1}^n \bm{y}_i  - \min_{j=1}^n \bm{y}_j  \geq \frac{1}{d_{\max}} \]
where $d_{\max}$ is the maximum generalized degree of graph. Now, consider $a > 0$ and $b< 0$ such that $a - b = \ell$, then we know $a^2 + b^2$ achieves minimum value at $\ell^2/2$. Using this we get that 
\[\left\{\max_{i=1}^n \bm{y}_i\right\}^2  + \left\{\min_{j=1}^n \bm{y}_j\right\}^2  \geq \frac{1}{2d^2_{\max}}\]
This gives us the lower bound $\Vert\bm{y}\Vert \geq \frac{1}{\sqrt{2} d_{\max}}.$ Putting this back in Eq.~ \eqref{eq:upper-lower-norm} we get 
$\beta^*\Vert\bm{J}\Vert \geq \frac{ \lambda_2^L}{\sqrt{2} d_{\max}}$.
Further, note that
$\Vert\bm{J}\Vert^2 = \sum_{v\neq \sink}\bm{J}_v^2 + \left(\sum_{v\neq \sink}\bm{J}_v\right)^2\leq 2 \left(\sum_{v\neq \sink}\bm{J}_v\right)^2.$ So, we get that 
\begin{equation}
\label{eq:L-bound-main}
  \beta^* \geq  \frac{\lambda_2^L}{2 d_{\max}\sum_{v\neq \sink}\bm{J}_v}.  
\end{equation}
\end{proof}

\section{Distributed Solver}
\label{sec:distributed_solver}
In this section, we will first present our distributed random walk based algorithm. Then, we will discuss the proof of our main function \ref{alg:rw_solver}, followed by an in-depth analysis of \ref{alg:rw_solver_compute} subroutine. Finally, we will discuss effective resistance computation using our solver.
\subsection{The Algorithm}
\label{subsec:alg_detail}
We now present the main distributed algorithm.
\begin{varalgorithm} {DRW-LSolve}
\caption{$(w,\bm{b},\epsilon,\kappa)$ Run by controller node $\sink$}
  {\label{alg:rw_solver}}
  \begin{algorithmic}[1] 
        \REQUIRE Parameters (i) $w$: Edge weight function (ii) $\pmb{b}$: RHS of $\pmb{x}^TL = \pmb{b}^T$ (iii) $\epsilon$: Error parameter (iv) $\kappa$: a granularity factor s.t. approximation guarantee holds only if $\kappa < \pot{i} d_i$.
        \STATE Send $\kappa,\epsilon,\sum_{v \in V} \vol{v}, \sum_{v\neq \sink}\bm{b}_v$ to all $u \in \nonsinks$
          \STATE $\bm{J}\leftarrow \bm{b}/\left(\sum_{v \neq \sink}\bm{b}_v\right)$ 
          \STATE For every $u \in V_s$ send $\bm{J}_u$ to $u$
          \STATE $\beta\leftarrow1$ \hfill {\small  /* End of Initializations */}
          \REPEAT \label{ln:repeat-start}
             {\STATE $\beta\leftarrow\beta/2$ \hfill {\small /* First value of $\beta$ is $1/2$ */}
                      \STATE {Send message ``Initiate \ref{alg:rw_solver_compute} with rate $\beta$'' to all nodes}
                        \STATE \textbf{WAIT} to receive $\heta{u}$ values from $u\in \nonsinks$ and then compute $\hemax\leftarrow \max_{u\in \nonsinks}\heta{u}$}
                \UNTIL{$\hemax< 3/4(1-\epsilon)$
                 \hfill{\small /* i.e., $\beta<3\beta^*/4$ */} \label{ln:repeat-end}
                        } 
            \STATE {Send message ``\ref{alg:rw_solver} over'' to all nodes}\label{ln:send-z}
              \end{algorithmic}
\end{varalgorithm}
\begin{varalgorithm} {DRW-LSolve-Slave}
  \caption{Run by nodes $u\in V\setminus\{\sink\}$}
  {\label{alg:rw_solver_other_nodes}}
  \begin{algorithmic}[1] 

        \STATE Receive parameters from the controller node $\sink$: Granularity factor $\kappa$, relative error $\epsilon$, $\sum_{v \in V} \vol{v}$, $\sum_{v\neq \sink}\bm{b}_v$, $\bm{J}_u$ if $u\in V_s$ 

              \IF {Message received ``Initiate \ref{alg:rw_solver_compute} with rate $\beta$'' from $\sink$}
                  {
                     \STATE Run \ref{alg:rw_solver_compute} $(\beta)$ to compute $\heta{u}$
                     \STATE Send $\heta{u}$ to controller node $\sink$
                  }
               \ELSIF {Message received ``\ref{alg:rw_solver} over''}
               {
               \IF{$\kappa<\frac{\heta{u}}{(1+\epsilon)}$}
            {\STATE $ \hpot{u} \leftarrow  \frac{\sum_{v\neq \sink}\bm{b}_v}{\beta}\left(\frac{\heta{u}}{d_u} \right)$}
            \ELSE
            {\STATE Send Message ``Estimate can't be guaranteed"
            \STATE $ \hpot{u} \leftarrow 0$
            }
            \ENDIF   
            \RETURN $\hpot{u}$
}
\ENDIF
              \end{algorithmic}
\end{varalgorithm}

\begin{varalgorithm} {DRW-Compute} \caption{$(\beta)$ Run by node $u\in \nonsinks$}
  {\label{alg:rw_solver_compute}}
  \begin{algorithmic}[1] 
           \STATE Initialize timer $T=0$, $Q_t(u)=0$, $cnt = 0$, and $\pmb{\eta}^{est}_0=0$
           \REPEAT
          {
                     \STATE $T\leftarrow T+1$
                     \IF {$u \in V_s$}
           \STATE Generate a data packet with probability $\beta\bm{J}_u$ and place in queue.
           \ENDIF
               \IF {$Q_t(u)$ is non-empty}
                    {\STATE $u$ picks a neighbor $v$ with probability $\probm{u,v}=\frac{w_{uv}}{\sum_{v:u\sim v}w_{uv}}$.
                      \STATE $v$ adds packet $p$ in $Q_{t+1}(v)$
                      \STATE $u$ deletes packet $p$ from $Q_{t+1}(u)$
        \STATE $cnt\leftarrow cnt+1$
       }
        \ENDIF   }  
        
        \STATE $\pmb{\eta}^{est}_T=\dfrac{cnt}{T}$
        \UNTIL{$0<\pmb{\eta}^{est}_{T}-\pmb{\eta}^{est}_{T-1}\leq\epsilon$}
           \STATE Send $\pmb{\eta}^{est}_{T}$ to controller node $\sink$ \hfill {\small /* Queue occupancy probability estimate $\heta{u}=\pmb{\eta}^{est}_{T}=\frac{cnt}{T}$  */}
              \end{algorithmic}
\end{varalgorithm}
Our main solver algorithm, \ref{alg:rw_solver}, is directed by a single node, the sink $\sink$, which we call as the \emph{controller node}. 
The algorithm takes as input $w,\pmb{b}$, $\epsilon$ and $\kappa$. The weight function $w$ specifies the Laplacian matrix $L$. The error of the solution to $\pmb{x}^T L = \pmb{b}^T$ is controlled by the parameter $\epsilon$. The fourth parameter $\kappa$ is a user-defined granularity factor which expresses the fact that our algorithm is not able to accurately approximate those coordinates of the solution vector whose value is very small. Specifically, our distributed solver returns an approximate solution $(\hat{\pmb{x}})$ to the given Laplacian system $\pmb{x}^TL=\pmb{b}^T$ with a relative error guarantee of $\epsilon$ {\em only} for those nodes $u \in \nonsinks$ for which $\kappa(1+\epsilon)<\pmb{x}_u d_u$ where $d_u$ is the generalized degree of node $u$.

The central loop of the main algorithm extends from Line~\ref{ln:repeat-start} to Line~\ref{ln:repeat-end} where it works downwards from $\beta = 1/2$ to find a value of $\beta < \beta^*$. To do so it asks all the nodes to run the subroutine \ref{alg:rw_solver_compute} with the current value of $\beta$. When the nodes have finished running \ref{alg:rw_solver_compute} they have computed an approximate version of their occupancy probability for the Data Collection problem with parameter $\beta$ and they send this back to the controller node, i.e., $\sink$. If the maximum occupancy probability received is below the threshold specified in Line~\ref{ln:repeat-end} then the controller node is satisfied that the last value of $\beta$ is at most $3\beta^*/4$ and it calls a halt to the algorithm. 

    


All nodes apart from the controller, $\sink$, run \ref{alg:rw_solver_other_nodes}. They simply initiate \ref{alg:rw_solver_compute} at the controller's direction with the appropriate value of $\beta$ and keep track of their queue occupancy using variable $cnt$. Finally, when the controller informs them that the algorithm is over if their queue occupancy probability estimate satisfies the condition $\kappa<\frac{\heta{u}}{(1+\epsilon)}$ where $\kappa$ is the user-defined granularity factor, then the nodes scale their values and return the answer. Otherwise, they return zero value as the estimate can't be guaranteed in such cases.

In the subroutine~\ref{alg:rw_solver_compute}, each node simply simulates the Data Collection process with the given parameter and sends the estimate of its occupancy probability to the controller after the allotted time is over.

\paragraph{Discussion about the model}

The GP-CONGEST model as discussed before, considers communication is done in discrete synchronous rounds, with each round taking in worst-case $d_{\max}$ time to simulate. In each such round, every node can send $O(\log n)$ size message to one of its neighbors while receiving multiple such messages from different neighbors. Also, it can perform local computations based on messages from previous rounds. To see the implementation of this model, note that whenever the subroutine \ref{alg:rw_solver_compute} is run it starts from scratch and runs for a number of time steps which is polynomial in $n$, specifically equal to the worst-case hitting time of random walk and in the inverse of the error parameters. In each such round, every node chooses a neighbor and sends a data packet to it. Since, these data packets are only used for queue occupancy estimation i.e., sender queue decrements while the receiver queue increments, so these packets need not to be distinguishable. However, once the estimation is done the estimate queue occupancy probability along with node id needs to be sent to a neighbor so that it finally reaches the controller node $\sink$. Node identifier by definition is $O(\log n)$ bits and the estimate would also require the same size as long as error parameters are polynomially bounded.

\subsection{Overall Analysis}
\ref{alg:rw_solver} works by repeatedly calling \ref{alg:rw_solver_compute} till it finds a stable data rate $3\beta^*/8\leq\beta<3\beta^*/4$ and computes the solution to be output at that value of $\beta$. Clearly the correctness of \ref{alg:rw_solver} depends on the ability of \ref{alg:rw_solver_compute} to return a good approximation of the occupancy probabilities involved in the given time. So, we need the following lemmas that characterize the behavior of \ref{alg:rw_solver_compute} above and below $\beta^*$. 
\begin{lemma}
\label{lem:stable_beta_comp}
Given $\epsilon>0$ and $\epsilon_1,\epsilon_2>0$ such that $\epsilon=\epsilon_1+\epsilon_2$ and a data rate $\beta<\beta^*$, \ref{alg:rw_solver_compute} $(\beta)$ returns an estimate $\heta{u}$ for all $u \in \nonsinks$ such that $|\heta{u}-\bm{\eta}_u|\leq (\epsilon_1+\epsilon_2)\bm{\eta}_u$ for $\kappa <\bm{\eta}_u$ in $\left(64\thit\log \epsilon_1^{-1}+\frac{4\log n}{\kappa^2\epsilon_2^2}\right)$ rounds of GP-CONGEST model where $\thit$ is the worst-case hitting time of random walk on the underlying graph.
\end{lemma}
Lemma~\ref{lem:stable_beta_comp} is a mixing time result for the Markov chain $Q_t$ which is guaranteed to be ergodic when $\beta < \beta^*$. However, a key insight used by \ref{alg:rw_solver} is that when $\beta \geq \beta^*$ there is bound to be one queue whose occupancy probability rises towards 1 if we wait long enough.
\begin{lemma}
\label{lem:unstable_beta_comp}
Given $\epsilon>0$ and $\epsilon_1,\epsilon_2>0$ such that $\epsilon=\epsilon_1+\epsilon_2$ and a data rate $\beta\geq\beta^*$, \ref{alg:rw_solver_compute} $(\beta)$ returns an estimate $\heta{u}$ for all $u \in \nonsinks$ such that $\hemax=\max_{u\in \nonsinks} \heta{u} \geq 1 - (\epsilon_1+\epsilon_2)$ in $\left(64\thit\log \epsilon_1^{-1}+\frac{4\log n}{\kappa^2\epsilon_2^2}\right)$ rounds of GP-CONGEST model where $\thit$ is the worst-case hitting time of random walk on the underlying graph.
\end{lemma}
\ref{alg:rw_solver} uses the insight of Lemma~\ref{lem:unstable_beta_comp} to revise the value of $\beta$ downwards by a factor of 2 if \ref{alg:rw_solver_compute} finds a queue which has a very high occupancy. If, on the other hand, there is no such queue, then Lemma~\ref{lem:stable_beta_comp} tells us that the occupancy probability vector is correctly approximated and we are done. Formally we can say the following.
\begin{lemma}
\label{lem:lsolve_correct}
\ref{alg:rw_solver}$(w,\bm{b},\epsilon,\kappa)$ returns a $1 \pm \epsilon$-approximate  solution to the Laplacian $\bm{x}^TL= \bm{b}^T$ in $\left(d_{\max}\left(64\thit\log \epsilon_1^{-1}+ \frac{4\log n}{\kappa^2\epsilon_2^2}\right)\log\frac{16d_{\max}}{3\lambda_2^L}\right)$
time where $\epsilon_1,\epsilon_2>0$ such that $\epsilon=\epsilon_1+\epsilon_2$, $\thit$ is the worst-case hitting time of random walk on the underlying graph, $d_{\max}$ is its generalized maximum degree, and $\lambda_2^L$ is the second smallest eigenvalue of the Laplacian of graph.
\end{lemma}
\begin{proof}[Proof of Lemma~\ref{lem:lsolve_correct}]
We will first prove the correctness of our solver \ref{alg:rw_solver} based on the parameters set for the computation. After that we will bound the time for computing those correct estimates.
\paragraph{Correctness}
As discussed before, our distributed solver works by choosing sink node $\sink$ as a controller of algorithm whose job is to coordinate all other nodes and compute the stable data rate for operation. For this it starts a binary search from $\beta=1/2$ and halves the data rate whenever it is found to be unstable. We operate at rate $3\beta^*/4$ such that we use Lemma~\ref{lem:unstable_beta_comp} as a test condition for instability i.e., if $\beta\geq 3\beta^*/4$ then $\hemax\geq 3/4(1-\epsilon)$ where $\hemax$ is computed using estimates $\heta{u}$ returned by all nodes $u \in \nonsinks$ using \ref{alg:rw_solver_compute} subroutine. This rate of operation works fine for our solver as we need to correctly identify unstable data rates $\beta \geq \beta^*$ and for all such data rates, since these are greater than $3\beta^*/4$ as well, we are able to correctly identify them. However, for rates $3\beta^*/4\leq \beta< \beta^*$, although these are stable our solver might indicate them as unstable. But, this is okay as we finally need a stable data rate for computation and by the definition of our binary search, the value at which it would stop i.e., $\hemax< 3/4(1-\epsilon)$ is at least $3\beta^*(1-\epsilon)/8$ which is also a stable data rate. Moreover, from Lemma~\ref{lem:stable_beta_comp} we know for stable data rates \ref{alg:rw_solver_compute} returns a $(1-\epsilon)$ estimate of queue occupancy probability which can then be used to return the solution to the original Laplacian equation $\bm{x}^TL= \bm{b}^T$ after appropriate scaling.


\paragraph{Time} Now, let us bound the time taken by the solver to compute the estimates. We know that the solver performs binary search from $\beta=1/2$ and halves the data rate each time it is found to be unstable. Also the least value it can reach is $3\beta^*(1-\epsilon)/8$. Let $\tt{I}_t$ be the number of such binary search iterations till it finds the stable data rate. In each such iteration, nodes $u\in \nonsinks$ in parallel use subroutine \ref{alg:rw_solver_compute} and there is a message exchange (parameters from the controller to others and $\heta{u}$ values from nodes to the controller) which can take maximum upto $2\diam$ time where $\diam$ is the diameter of the graph. So from Lemma~\ref{lem:stable_beta_comp}, given $\epsilon_1,\epsilon_2>0$ such that $\epsilon=\epsilon_1+\epsilon_2$, the number of distributed rounds of GP-CONGEST model required for one binary search iteration is 
$\left(64\thit\log \epsilon_1^{-1}+\frac{4 \log n}{\kappa^2\epsilon_2^2}+2\diam\right)$. 
Since the binary search begins at $\beta = 1/2$ and ends above $3\beta^*(1-\epsilon)/8$, we know that the number of iterations $\tt{I}_t\leq\log (4/3\beta^*(1-\epsilon))$. If $\epsilon<1/2$, then $\tt{I}_t \leq \log (8/3\beta^*)$ and from Lemma \ref{lem:beta_bound} we know $\beta^*\geq \frac{\lambda_2^L}{2d_{\max}}$ as $\sum_{v\neq \sink}\bm{J}_v=1$. Also, by the definition of our model we know each distributed round takes in worst-case $d_{\max}$ time. So, the overall running time for \ref{alg:rw_solver}$(w,\bm{b},\epsilon,\kappa)$ is
\begin{equation}
\left(d_{\max}\left(64\thit\log \epsilon_1^{-1}+ \frac{4\log n}{\kappa^2\epsilon_2^2}+2\diam\right)\log\frac{16d_{\max}}{3\lambda_2^L}\right).
\label{eq:k_iter_d_reff_time}
\end{equation}
As the diameter of graph cannot be greater than the worst-case hitting time of random walk on it, $\diam$ term in the running time result can be suppressed.
\end{proof}

Note that although the given approximate bound as mentioned in Lemma~\ref{lem:stable_beta_comp} is returned only for nodes which satisfy $\kappa<\bm{\eta}_u$. However, since in the algorithm (\ref{alg:rw_solver_other_nodes}) we can only keep track of the estimate $\hat{\pmb{\eta}}$, so we check the condition $\kappa<\frac{\heta{u}}{(1+\epsilon)}$. This changes our overall error guarantee condition to $\kappa(1+\epsilon)<\bm{\eta}_u$.

Now, given error parameter $\epsilon>0$ such that $\epsilon=\epsilon_1+\epsilon_2$, let us set $\epsilon_1=\epsilon/4$ and $\epsilon_2=3\epsilon/4$. Using these values and the proofs of correctness of the subroutine \ref{alg:rw_solver_compute} (Lemma~\ref{lem:stable_beta_comp} and Lemma~\ref{lem:unstable_beta_comp}) and the main function \ref{alg:rw_solver} (Lemma~\ref{lem:lsolve_correct}) together prove Theorem~\ref{thm:rand_walk_solver}. So, we now turn to the analysis of \ref{alg:rw_solver_compute} which is the heart of our method. 

\subsection{Analyzing \ref{alg:rw_solver_compute}}
\label{subsec:rw_comp_correctness}

\subsubsection{Analyzing \ref{alg:rw_solver_compute} when $\beta<\beta^*$} For the case of stable data rates $\beta<\beta^*$, from Theorem~\ref{thm:gillani-stability} we know that for such rates the Markov chain $Q_t$ defined on the queue size of nodes is ergodic and has a stationary distribution. To estimate steady-state queue occupancy probabilities we need to be close to stationarity so we first bound the mixing time for $Q_t$.
\begin{lemma}
\label{lem:stationrity_time}
Given an irreducible and aperiodic Markov chain $Q_t$ described by the data collection process defined on $(\nn0)^{|V|-1}$ having transition matrix $\probm{\cdot,\cdot}$ and a stationary distribution $\pi$. For $\beta<\beta^*$ we have $||\mathcal{P}^t(0,\cdot)-\pi||_{TV} \leq \epsilon^*$ for
$t=64\thit\log \epsilon^{*-1}$,
where $\thit$ is the worst-case hitting time of random walk on the underlying graph.
\end{lemma}
\begin{proof}[Proof of Lemma~\ref{lem:stationrity_time}]
We first note that our Markov chain $Q_t$ is {\em stochastically ordered} (c.f.~\cite{Lund-MOR:1996}). To understand what this means we define a natural partial order on $\nn0^{|\nonsinks|}$ as follows: $\bm{x} \preceq \bm{y}$ if $\bm{x}_v \leq \bm{y}_v$ for all $v \in V$. A function $f : \nn0^{|V|} \rightarrow \mathbb{R}$ is said to be increasing if $\bm{x} \preceq \bm{y}$ implies that $f(\bm{x}) \leq f(\bm{y})$. Given two random processes $X$ and $Y$ supported on $\nn0^{|\nonsinks|}$ we say $X$ is stochastically dominated by $Y$ if $\ex{f(X)}\leq \ex{f(Y)}$ for every increasing function $f$. We now state the stochastic orderedness property as a claim.
\begin{claim}
\label{clm:stochastic-ordering}
Given two instances of the data collection process $Q_t$ and $Q'_t$ such that $Q_0 \preceq Q'_0$, $Q_t$ is stochastically dominated by $Q'_t, \forall t \geq 0$. In particular, this means that $\prob{Q_t(v) > 0} \leq \prob{Q'_t(v) > 0}$ for all $v \in \nonsinks$.
\end{claim}
The proof of this claim follows by constructing a coupling between the two chains such that each of them performs exactly the same transmission actions. In case one of the chains is empty then the transmission action is a dummy action. It is easy to see that the stochastic ordering follows naturally for the data collection chain. 

To use this claim, for our irreducible and aperiodic Markov chain $Q_t$ described by the data collection process defined on $(\nn0)^{|V|-1}$ having transition matrix $\probm{\cdot,\cdot}$ and a stationary distribution $\pi$, let us define another irreducible and aperiodic Markov chain $\bar{Q}_t$ with state space $(\nn0)^{|V|-1}$ which has already achieved stationary distribution $\pi$.

Now, consider the coupling $(Q_t,\bar{Q}_t)$ on $(\nn0)^{|V|-1} \times (\nn0)^{|V|-1}$ defined over random sequences $\{0,1\} \times \{\prod_{v \in \nonsinks} \Gamma(v)\} $ where $\Gamma(v)$ is the set of one-step destinations from node $v$, such that the chain $Q_t$ starts with empty queues i.e., $Q_0(u)=0, \forall u \in \nonsinks$ and the queues in chain $\bar{Q}_t$ are populated according to stationary distribution $\pi$ i.e., $\bar{Q}_0(u), \forall u \in \nonsinks$ is non-zero in general. Such Markov chains are said to be stochastically ordered chains in the queueing theory and have a property that the Markov chain which dominates the other chain in the beginning, will always maintain dominance over it for
all $t>0$. 

Now, under this coupling we allow the two chains to run in a way that any data generation or data transmission decision made by any queue in one chain is followed by the corresponding queue in the other chain as well. However, suppose the packets already in $\bar{Q}_t$ are distinguished from the newly generated packets and the latter ones get a preference in the transmission. Given that $\bm{\eta}^t_u$ is the queue occupancy probability of node $u$ in the chain $Q_t$ and $\bm{\eta}_u$ is its steady-state queue occupancy probability in Markov chain $\bar{Q}_t$. Then, we know $\bm{\eta}^t_u \leq \bm{\eta}_u$ from Claim~\ref{clm:stochastic-ordering}. To ensure both chains get coupled all the old packets in $\bar{Q}_t$ need to be sunk. However, by our preference in transmission, the probability that such packets move out of queue in one time step is equal to the probability that corresponding queue in $Q_t$ is empty i.e., $1-\bm{\eta}_u^t$. Also, we have  $1-\bm{\eta}_u^t\geq 1-\bm{\eta}_u\geq \min_u 1-\bm{\eta}_u\geq 1-\bm{\eta}_{max}$. So, the number of packets that move out of the queue in $\bar{Q}_t$ in $t_{exp}$ time steps is at least $t_{exp}(1-\bm{\eta}_{max})$. Now, since we know the sink collects data packets at rate $\beta$, so in $t_{exp}$ time steps we have $\beta t_{exp} (1-\bm{\eta}_{max})$.

Now if we consider the expected number of packets already residing in the queues of Markov chain $\bar{Q}_t$ in the beginning (as per the stationary distribution), we know from Little's law \cite{Gross-BOOK:2008} it is equal to the product of expected latency of a data packet to reach the sink and the rate of collection by the sink. To bound the expected latency we have the following claim.

\begin{claim}
\label{clm:exp_latency}
Given a data collection process on a graph $G=(V,E,w)$ such that $\beta < \beta^*$, the expected time spent in $G$ by a data packet before it gets sunk is $\frac{\thit}{1-\bm{\eta}_{max}}$ where $\thit$ is the worst-case hitting time of random walk on graph $G$ and $\bm{\eta}_{max}$ is the maximum queue occupancy probability of all nodes in $\nonsinks$  at stationarity.
\end{claim}

The proof of this claim follows by first analyzing the time taken by a data packet to hit the sink without any queueing delays which in the worst-case is $\thit$ and then combining it with the delay which is at most $\bm{\eta}_{max}$.

Now, to use this claim we know that the sink collects data packets at rate $\beta$, so in $t_{exp}$ time steps we have $\beta t_{exp} (1-\bm{\eta}_{max})=\frac{\beta\thit}{1-\bm{\eta}_{max}}$ which gives us $t_{exp}=\frac{\thit}{(1-\bm{\eta}_{max})^2}$ where $t_{exp}$ is the time by which all old packets in $\bar{Q}_t$ have sunk i.e., the expected time by which the two chains couple. 

To bound this time, we cannot operate very close to the critical data rate $\beta^*$, so we operate at $3\beta^*/4$, a data rate which is constant factor away from the critical data rate. To use this value, let us first prove an important property of $\bm{\eta}$ with respect to data rate $\beta$. Given a $\beta'<\beta^*$, we have from Eq.~\eqref{eq:steady_state_queue_eq_gen} $\bm{\eta}'^T(I-\mathcal{P})=\beta'\bm{J}^T$. Multiplying both sides of this equation by $\beta^*/\beta'$ we have $\frac{\beta^*}{\beta'}\bm{\eta}'^T(I-\mathcal{P})=\beta^*\bm{J}^T$. This gives us $\bm{\eta}^*=\frac{\beta^*}{\beta'}\bm{\eta}'$ i.e., $\bm{\eta}$ is linear in $\beta$. From this property and the fact that at $\beta^*$ the maximum queue occupancy probability $\bm{\eta}_{max}=1$, operating at a data rate of $3\beta^*/4$ gives us $\eta_{max}=3/4$ which means $t_{exp}=16 \thit$.

The distance between two Markov chains is known to be related to their expected coupling time. We use the formulation of this property as presented by Levin, Peres and Wilmer~\cite{Levin-BOOK:2009}.
\begin{lemma}[Theorem 5.2, Levin et al. \cite{Levin-BOOK:2009}]\
\label{lem:levin_coup}
Let $\{(X_t,Y_t)\}$ be a coupling of Markov chains with transition matrix $\mathcal{P}$ with initial states $X_0=x$ and $Y_0=y$ and coupling time defined as $\tau_{couple}:=\min \{t:X_s=Y_s \text{ for all } s\geq t\}$, then,
$$||\mathcal{P}^t(x,\cdot)-\mathcal{P}^t(y,\cdot)||_{TV}\leq P_{x,y}\{\tau_{couple}>t\}. $$
\end{lemma}
For our given Markov chains, let $\bar{d}(t)=||\mathcal{P}^t(x,\cdot)-\mathcal{P}^t(y,\cdot)||_{TV}$. We know the expected coupling time of two chains is
$E_{x,y}(\tau_{couple})=16 \thit$. So, from Lemma~\ref{lem:levin_coup} and Markov's inequality we have 
\begin{align}
\bar{d}(t)&\leq \frac{E_{x,y}(\tau_{couple})}{t}=\frac{16 \thit}{t}.
\label{eq:stat_time}
\end{align}
Now analogous to mixing time, let us define $t_{m}(\epsilon):=\min\{t:\bar{d}(t) \leq \epsilon\}$ and $t_{m}=t_{m}(1/4)$ such that for $\epsilon^*=1/4$ we have $t_{m}(\epsilon^*)\leq t_{m}\log \epsilon^{*-1}$. So to bound the time by which the two Markov chains $Q_t$ and $\bar{Q}_t$ are $\epsilon^*$ away from each other i.e., $\bar{d}(t) \leq \epsilon^*$, from \eqref{eq:stat_time} we have 
\begin{equation}
t_{m}(\epsilon)\leq t_{m}\log \epsilon^{*-1} \leq 64\thit\log \epsilon^{*-1}.
\label{eq:t1_time}    
\end{equation}
\label{comm:tm_correction}
\end{proof}

Now, let us use this lemma to prove the correctness of our subroutine \ref{alg:rw_solver_compute} for $\beta<\beta^*$.

\begin{proof}[Proof of Lemma~\ref{lem:stable_beta_comp}]
\label{comm:lemma3_proof}
Given a data rate $\beta<\beta^*$, our subroutine computes the steady-state queue occupancy probability of nodes by first running the Markov chain $Q_t$ defined on the queue size of nodes $u \in \nonsinks$ close to the stationary distribution. After the chain is close to its stationarity, then it starts sampling the values of its queue occupancy.


Firstly, let us suppose Markov chain $Q_t$ is $\epsilon_1$ distance from the stationary distribution $\pi$. Using Lemma~\ref{lem:stationrity_time} we can bound this time, let this be $t_1$. Now, given that all nodes are initially empty for $||\mathcal{P}^t(0,\cdot)-\pi||_{TV} \leq \epsilon_1$ we have $t_1=64\thit\log \epsilon_1^{-1}$.
So after $t_1$ time, as we are close to the steady-state distribution we can start sampling values to get an estimate of the steady-state queue occupancy probabilities. For this, let us consider a node $u \in \nonsinks$ and define variables $X_i(u)=1_{\{Q_i(u)>0\}}$, $Z_t(u)=\sum_{i=1}^t X_i(u)-\ex{\sum_{i=1}^t X_i(u)}$, and $\bar{X}(u)=\dfrac{1}{t'_u}\sum_{i=1}^{t'_u}X_i(u)$ where $t'_u$ is the sampling time for node $u$. Since, we know by time $t_1$ the Markov chain is $\epsilon_1$ close to the stationarity, so we have $\left|\ex{\bar{X}(u)}-\bm{\eta}_u\right|\leq \bm{\eta}_u \epsilon_1$ or $\left|\ex{\sum_{i=1}^{t'_u} X_i(u)}-{t'_u}\bm{\eta}_u\right|\leq {t'_u}\bm{\eta}_u \epsilon_1$ where $\bm{\eta}_u$ is the steady-state queue occupancy probability of node $u$. Moreover, as $\ex{Z_t(u)}=0$ so the sequence $\left(Z_t(u)\right)_{t>0}$ clearly forms a martingale and by definition we have $\left| Z_t(u)-Z_{t-1}(u)\right|=\left|X_t(u)-\ex{X_t(u)}\right|<1$. For such martingales with bounded differences we have the following concentration result.
\begin{theorem}[Azuma-Hoeffding inequality \cite{Hoeffding-JSTOR:1994*}]
\label{thm:azuma-hoeffding}
Let us consider a martingale $\left\{X_t:t=0,1,2,\cdots\right\}$ such that $\left|X_t-X_{t-1}\right|<c_t$ almost surely, where $c_t>0$ is a constant. Then for all positive reals $\epsilon^*$,
\begin{equation}
\prob{\left|X_n-X_0\right|\geq \epsilon^*}\leq 2 \exp \left(\frac{-2\epsilon^{*2}}{\sum_{t=1}^n c_t^2} \right)
\label{eq:azuma-hoeffding}
\end{equation}
\end{theorem}

Now, given a constant $\epsilon_2>0$ we have ${t'_u}\bm{\eta}_u(1-\epsilon_1)\epsilon_2\geq0$. Now, using Azuma-Hoeffding's inequality (Theorem~\ref{thm:azuma-hoeffding}) for martingale $Z_t(u)$, we know $Z_0(u)=0$ so we have
\begin{align}
&\prob{\left\lvert\sum_{i=1}^{t'_u} X_i(u)-\ex{\sum_{i=1}^{t'_u} X_i(u)}\right\rvert
\geq {t'_u}\bm{\eta}_u(1-\epsilon_1)\epsilon_2}\nonumber\\
&\leq 2\exp{\left(\dfrac{-2t'_u{}^2(\bm{\eta}_u(1-\epsilon_1)\epsilon_2)^2}{t'_u}\right)}\nonumber\\
&\leq 2\exp{\left(-2(t'_u\bm{\eta}_u^2(1-\epsilon_1)^2\epsilon_2^2)\right)}.
\label{eq:hoeff_tail_bound}
\end{align}
Above equation (Eq.~\eqref{eq:hoeff_tail_bound}) can be rewritten as follows.
\begin{align}
&\prob{\left\lvert\bar{X}(u)-\ex{\bar{X}(u)}\right\rvert
\geq \bm{\eta}_u(1-\epsilon_1)\epsilon_2}\nonumber\\
&\leq 2\exp{\left(-2(t'_u\bm{\eta}_u^2(1-\epsilon_1)^2\epsilon_2^2)\right)}.
\label{eq:hoeff_tail_bound_2}
\end{align}

So, after sampling for time $t'_u$ with high probability we get $(1-\epsilon_1-\epsilon_2)$ estimate of steady-state queue occupancy probability $\bm{\eta}_u$ when 
$$t'_u=\dfrac{\log n}{\bm{\eta}_u^2(1-\epsilon_1)^2\epsilon_2^2}.$$
or $t'_u$ is greater than the given term. To bound this term we know $\epsilon_1 \leq 1/2$. So, we have $t'_u=\frac{4\log n}{\bm{\eta}_u^2\epsilon_2^2}$. Moreover, as we are sampling component wise (separately for each node $u \in\nonsinks$), so our overall sampling time will be
\begin{equation}
t_2=\max_u t'_u\geq\max_u\dfrac{4\log n}{\bm{\eta}_u^2\epsilon_2^2}.
\label{eq:t2value}
\end{equation}
So, we get the overall time as 
$t_1+t_2\geq 64\thit\log \epsilon_1^{-1}+\max_u \left(\frac{4\log n}{\bm{\eta}_u^2\epsilon_2^2}\right)$. Now, to bound this time we consider all nodes $u \in \nonsinks$ such that their node potentials satisfy $\kappa<\bm{\eta}_u$ where $0<\kappa<1$.
So, using this fact our subroutine computes the estimates $\heta{u}, \forall u \in \nonsinks$ such that \begin{equation}
|\heta{u}-\bm{\eta}_u| \leq (\epsilon_1+\epsilon_2)\bm{\eta}_u, \forall \kappa<\bm{\eta}_u 
\label{eq:queue_occup_estimate}
\end{equation} 
in time
 $\left(64\thit\log \epsilon_1^{-1}+\frac{4\log n}{\kappa^2\epsilon_2^2}\right)$  where $\epsilon>0$ such that $\epsilon=\epsilon_1+\epsilon_2$.
\end{proof}

\subsubsection{Analyzing \ref{alg:rw_solver_compute} for $\beta\geq\beta^*$}
For the case of unstable data rates $\beta \geq \beta^*$, we know the Markov chain is non-ergodic. However, our subroutine can still return the values which can serve as an indicator for identifying such data rates. In particular, we prove the following result.

\begin{proof}[Proof of Lemma~\ref{lem:unstable_beta_comp}]
To understand the working of our subroutine for data rates $\beta\geq \beta^*$, 
let us first consider the following property of queue occupancy probability.
\begin{claim}
\label{clm:eta_incr_beta}
Given a data collection process with source nodes having an independent Bernoulli data generation with rate $\beta$. Let $Q^\beta_t$
represent the queues at time $t$ for all nodes $u \in \nonsinks$. Then, for all such nodes $\prob{Q^\beta_t(u)>0}$ is 
an increasing function of $\beta$.
\end{claim}
This claim can be proved by first showing the stochastic orderedness property of Markov chain $Q_t^\beta$ (see Claim~\ref{clm:stochastic-ordering}) and then proving the monotonicity in $\beta$ using a similar coupling. 

Now, to use this claim let us consider a data rate close to the critical rate $\beta^*$ i.e., $\beta'=\beta^*-\epsilon^*$. Since $\beta'<\beta^*$ from Claim~\ref{clm:eta_incr_beta} we have $\prob{Q^{\beta'}_t(u)>0}<\prob{Q^{\beta^*}_t(u)>0}$ i.e., $\heta{u}^{\beta'}<\heta{u}^{\beta^*}$. As, $\beta'$ is stable we know from Lemma~\ref{lem:stable_beta_comp}, our subroutine will return estimates $\heta{u}^{\beta'}$ in time  $64\thit \log \epsilon_1^{-1}+\frac{4 \log n}{\kappa^2\epsilon_2^2}$ where $\thit$ is the worst-case hitting time of random walk on the underlying graph. So, after the given time we can compute $\hemax^{\beta'}=\max_{u\in \nonsinks}\heta{u}^{\beta'}=1-(\epsilon_1+\epsilon_2+\epsilon^*)$. Now, since $\heta{u}^{\beta'}<\heta{u}^{\beta^*}$ and in fact $\bm{\eta}$ is linear in $\beta$, in same time for $\beta\geq \beta^*$ solver will return $\heta{u}^{\beta^*}$ such that $\hemax^{\beta^*}=\max_{u\in \nonsinks} \heta{u}^{\beta^*} \geq 1 - (\epsilon_1+\epsilon_2)$, where error parameter $\epsilon=\epsilon_1+\epsilon_2$.
\end{proof}
 \emph{Remark:} As discussed before, instead of operating at data rates very close to the critical data rate, we opt for rate $3\beta^*/4$. So for all data rates $\beta \geq 3\beta^*/4$ from Lemma~\ref{lem:unstable_beta_comp} we have $\hemax \geq 3/4(1 - \epsilon)$. This acts as a test condition for checking the stability of data rates which is used by our solver \ref{alg:rw_solver} to compute a stable data rate for computation operation as discussed in Section \ref{subsec:alg_detail}.

\subsection{Effective Resistance Computation} \ref{alg:rw_solver} can be directly used to calculate the effective resistance $\reff{u,v}$ between a pair of nodes  $u,v\in V$. We prove that result next. 
\begin{proof}[Proof of Theorem~\ref{thm:eff_res_compt}]
We know from the steady-state equation (Eq.~\eqref{eq:steady-state_pot}), the effective resistance between two nodes $u$ and $v$ is $\reff{u,v}=\frac{\pot{u}-\pot{v}}{\beta}$. To compute this resistance $\reff{u,v}$ let us assume $u$ is the source node (denoted by $s$) and $v$ is the sink (denoted by $\sink$). Since, we always assume queue occupancy probability of sink to be zero so to compute the effective resistance we only need to compute $\bm{\eta}_s$ accurately. Using Theorem~\ref{thm:rand_walk_solver} we know \ref{alg:rw_solver} can compute estimates $\heta{u}$ such that $|\heta{u}-\bm{\eta}_u| \leq (\epsilon_1+\epsilon_2)\bm{\eta}_u, \forall \kappa<\bm{\eta}_u$ in $\left(64\thit\log \epsilon_1^{-1}+ \frac{4\log n}{\kappa^2\epsilon_2^2}+2\diam\right)\log\frac{16d_{\max}}{3\lambda_2^L}$ rounds. In this case, since $u$ is the only source and $v$ is the sink, the expected latency is $\thit(u,v)$ instead of $\thit$. Also, we can ignore the $\diam$ factor in running time which comes from the exchange of messages between the controller node and other nodes, by making the source node $u$ as the controller. In such a scenario, for each binary search iteration source node will need to run \ref{alg:rw_solver_compute} locally as well.

Now, to accurately compute $\bm{\eta}_s$ we need to set value of $\kappa$ appropriately. For this consider the harmonic property of potential $\bm{x}$. By this property, we know that the potential at source node $\pot{s}$ is the maximum. So, we have $\pot{s}>\pot{u}, \forall u$. As, $\pot{u}=\eta_u/\vol{u}$, we have $\frac{\eta_s}{\vol{s}}>\frac{\eta_u}{\vol{u}}$ for all $u$. This gives
\begin{equation}
 \eta_s>\frac{\vol{s}}{\vol{u}}\eta_u>\frac{d_{\min}}{d_{\max}}\eta_u.
\label{eq:eta_s_u_rel}
\end{equation}
Moreover, since the lowest value of $\beta$ that \ref{alg:rw_solver} can reach through binary search is $3\beta^*(1-\epsilon)/8$, where $\epsilon=\epsilon_1+\epsilon_2$ so at that rate there exists a node $u^*$ whose queue occupancy probability is maximum i.e., $3(1-\epsilon)/8$.  Now, if we consider $\epsilon<1/2$, we know $\eta_{u^*}>3/16$. Using this value in Eq.~\eqref{eq:eta_s_u_rel} we have $\eta_s>\frac{3d_{\min}}{16d_{\max}}$. So, for correct estimation of $\eta_s$ we can set the value of $\kappa$ as $\frac{3d_{\min}}{16d_{\max}}$. Also, given $\epsilon=\epsilon_1+\epsilon_2$, let us set $\epsilon_1=\epsilon/4$ and $\epsilon_2=3\epsilon/4$. So, with each distributed round taking in worst-case $d_{\max}$ time, we get the overall time for effective resistance computation $\reff{u,v}$ as $\left(d_{\max}\left(64\thit(u,v)\log \frac{4}{\epsilon}+ \left(\frac{16d_{\max}}{3d_{\min}}\right)^2\frac{64\log n}{9\epsilon^2}\right)\log\frac{16d_{\max}}{3\lambda_2^L}\right)$. 
\end{proof}



\subsection{Obtaining Canonical Solution of the One-sink Laplacian System}
\label{subsec:caonical_mapping} 

As discussed before, $L$ is a singular matrix with the subspace generated by $\bm{1}$ being its null space. So we normally expect to find a solution $\bm{y}$ to have the property that $\langle \bm{y}, \bm{1}\rangle = 0$, i.e., $\sum_{i=1}^n \bm{y}_i = 0$. But clearly $\bm{x}$ obtained from the steady state of the Data Collection process has all coordinates non-negative, and so it is not in the canonical form.

However, we can write any other solution to Eq.~\eqref{eq:vector_queue_eq_1} as $\bm{\hat{\eta}} = \bm{\eta} + z \nu$ where $\nu$ is the stationary distribution of the Markov Chain with transition matrix $\mathcal{P}$, i.e., $\nu=\nu \mathcal{P}$ and $z\in \mathbb{R}$ is any constant. This is because we know $z\nu (I - \mathcal{P}) = 0$, so we have, $(\bm{\eta} + z \nu) (I - \mathcal{P}) = \bm{\eta}(I-\mathcal{P})$. So, by choosing the appropriate constant offset we can get to any other solution of the given Laplacian equation. In particular, for obtaining canonical solution to the one-sink Laplacian system from the stationary state of the Data Collection process, we can 
compute an offest value $z^*$ such that $\left\langle (\pmb{\eta}D^{-1} + z^*\cdot \pmb{1}), \pmb{1}\right\rangle = 0$ and then return the solution vector as $\sum_{i\neq v}\pmb{b}_i[(\pmb{\eta}D^{-1} + z^*\cdot \pmb{1})]/\beta$. However, note that by shifting the solution by $z^*$ it results in an additive error in addition to the relative error as guaranteed by Theorem~\ref{thm:rand_walk_solver}. The error guarantees in this case are similar to Theorem~\ref{thm:rand_walk_solver_gen} that we will discuss in the next section.

\section{Extension to General Laplacian Systems}
\label{sec:general_solver}
In this section, we present the extension of our solver from the \emph{one-sink} setting to a general setting. In particular, we propose a method of solving general Laplacian equation $\pmb{x}^TL=\pmb{b}^T$ where instead of one there are $\ell$ negative entries in $\pmb{b}$ vector.
\begin{figure}[ht!]
\centering
\begin{mdframed}
\begin{enumerate}
\item Given a general Laplacian equation $\pmb{x}^TL=\pmb{b}^T$ corresponding to graph $G=(V,E,w)$. Let $V_+=\{v:\pmb{b}_v>0\}$ and $V_-=\{v:\pmb{b}_v<0\}$.
\item Let $V_-=\{v_1,v_2,\cdots,v_{\ell}\}$ and split the Laplacian equation into $\ell$ one-sink Laplacian systems.
\begin{itemize}
    \item Let $(\pmb{x}^i)^TL=(\pmb{b}^{i})^T$ represent the one-sink Laplacian system where $v_i$ is designated as the sink node.
    \item Set $\pmb{b}^i_{v_i}=\pmb{b}_{v_i}$ for the sink node, $\pmb{b}^i_{u}=\frac{\pmb{b}_{v_i}}{\sum_{i=1}^{\ell} \pmb{b}_{v_i}}\pmb{b}_u$ for $u \in V_+$, and $\pmb{b}^i_{u}=0$ for $u\in V\setminus\{V_+,v_i\}$.
    \item Solve the one-sink system (see Figure~\ref{fig:schematic_steps} for the detailed steps) and let $\pmb{x}^i$ be the returned solution.
\end{itemize}
\item Return the overall solution
$\pmb{x}=\sum_{i=1}^{\ell}  \pmb{x}^i$.
\end{enumerate}
\end{mdframed}
\caption{Steps for computing solution to general Laplacian system of equations by splitting it into multiple one-sink Laplacian systems.}
\label{fig:schematic_steps_gen_solver}
\end{figure}

Our approach is based on splitting the general Laplacian equation into multiple one-sink Laplacian systems which can in turn be solved using algorithm discussed in Section~\ref{subsec:alg_detail}. We summarize the steps involved in this splitting in Figure~\ref{fig:schematic_steps_gen_solver} and propose a distributed algorithm \ref{alg:rw_solver_gen} for the problem. This algorithm is run by a general controller node which is chosen one among the nodes with negative $\pmb{b}$ value.
With this algorithm we can prove the following result.

\begin{varalgorithm} {GEN-DRW-LSolve}
\caption{$(w,\bm{b},\epsilon,\kappa)$ Run by general controller node }
  {\label{alg:rw_solver_gen}}
  \begin{algorithmic}[1] 
         \REQUIRE Parameters (i) $w$: Edge weight function (ii) $\pmb{b}$: RHS of $\pmb{x}^TL = \pmb{b}^T$ with $\ell$ negative entries (iii) $\epsilon$: Error parameter (iv) $\kappa$: a granularity factor s.t. approximation guarantee holds only if $\kappa < \pot{i} d_i$.
         \STATE Let $V_+=\{v:\pmb{b}_v>0\}$, $V_-=\{v:\pmb{b}_v<0\}$, and $\ell=|V_-|$
         \FOR{ each $v_i\in V_-$}
            \STATE Set $v_i$ as the sink node (controller) for the $i$th one-sink system 
            \STATE Set $\pmb{b}^i_{v_i}=\pmb{b}_{v_i}$ for the sink node, $\pmb{b}^i_{u}=\frac{\pmb{b}_{v_i}}{\sum_{i=1}^{\ell} \pmb{b}_{v_i}}\pmb{b}_u$ for $u \in V_+$, and $\pmb{b}^i_{u}=0$ for $u\in V\setminus\{V_+,v_i\}$
        \ENDFOR \hfill {\small  /* Split general Laplacian system into $\ell$ one-sink systems $(\pmb{x}^i)^TL=(\pmb{b}^{i})^T$*/}
        \STATE Send $\kappa,\epsilon,\pmb{b}^i$ to the sink nodes of $\ell$ one-sink systems i.e., $v_i \in V_-$
          \STATE Ask all $v_i\in V_-$ to run in parallel \ref{alg:rw_solver}$(w,\pmb{b}^i,\epsilon,\frac{\kappa}{\ell})$
          \STATE Once all parallel systems have converged, ask all nodes $u \in V$ to compute $\pmb{x}_u=\sum_{i=1}^{\ell}  \pmb{x}^i_u$
          \end{algorithmic}
\end{varalgorithm}
\begin{theorem}[General Distributed Solver]
\label{thm:rand_walk_solver_gen}
Given an undirected connected graph $G = (V,E,w)$ and $\bm{b} \in \mathbb{R}^{|V|}$ such that $\ell$ of its entries are negative, the corresponding Laplacian system is
\begin{equation}
\label{eq:solver-laplacian-system_gen}
\bm{x}^TL = \bm{b}^T.
\end{equation}
For a user-defined granularity factor $\kappa \in (0,1)$ and polynomially bounded error parameter $\epsilon\in (0,1)$ such that $\epsilon<1/2$, there is a distributed solver  \ref{alg:rw_solver_gen}$(w,\bm{b},\epsilon,\kappa)$ working in the GP-CONGEST model that takes 
$O\left(\ell d_{\max}\left(\thit\log \epsilon^{-1}+ \frac{\log n}{\kappa^2\epsilon^2}\right)\log\frac{d_{\max}}{\lambda_2^L}\right)$
time, to produce a vector $\bm{\hat{x}}$ that is an approximate solution to Eq.~\eqref{eq:solver-laplacian-system_gen}, where $d_{\max} = \max_{u\in V} d_u$, $\lambda_2^L$ is the second smallest eigenvalue of the Laplacian matrix $L$, and $\thit$ is the worst-case hitting time of the natural random walk on $G$.

Further, $\bm{\hat{x}}$ has the properties that
(i) $\hpot{v}$ is the estimate at node $v \in V$,
(ii) $\hpot{i} > 0$, $1 \leq i \leq n-1$, $\hpot{n} = 0$, and,
(iii) if $\bm{x}$ is an exact solution to Eq.~\eqref{eq:solver-laplacian-system} such that $\pot{i} > 0$, $1 \leq i \leq n-1$, $\pot{n} = 0$, 
then, $|\hpot{i} - \pot{i}| \leq \epsilon\pot{i}+\kappa(1+\epsilon)$.
\end{theorem}
\begin{proof}[Proof of Theorem~\ref{thm:rand_walk_solver_gen}] Given a general Laplacian equation $\pmb{x}^TL=\pmb{b}^T$ such that there are $\ell$ negative entries in $\pmb{b}$ vector. \ref{alg:rw_solver_gen} first splits this Laplacian equation into $\ell$ one-sink Laplacian systems $(\pmb{x}^i)^TL=(\pmb{b}^{i})^T$ using the steps defined in Figure~\ref{fig:schematic_steps_gen_solver}. After that it runs $\ell$ parallel instances of the \ref{alg:rw_solver}$(w,\pmb{b}^i,\epsilon,\frac{\kappa}{\ell})$ algorithm wherein it simulates the \emph{Data Collection process} with different rate vectors according to vector $\pmb{b}^i$. In particular, to distinguish between different instances of the Data Collection process we assume different colored data packets being generated and separate queues for them at all relay nodes. Note that for each one-sink system we choose one node with negative $\pmb{b}$ value as the sink and set all other nodes with negative value to zero, i.e., they act as relay nodes for other systems. To allow parallel operation of $\ell$ instances of the Data Collection process we use time-division multiplexing (TDM) and divide a time slot into $\ell$ mini slots, one for each instance. As a result, simulation of $\ell$ parallel one-sink Laplacian system slows them individually down by a factor of $\ell$. Now, using Lemma~\ref{lem:lsolve_correct} for each of the one-sink system and setting $\epsilon_1=\epsilon/4$ and $\epsilon_2=3\epsilon/4$ we get the desired result.

Note that for the approximation bound, we know from Theorem~\ref{thm:rand_walk_solver} for each one-sink system $(\pmb{x}^i)^TL=(\pmb{b}^{i})^T$ given $\hat{\pmb{x}}^i$ is the returned approximate solution by the solver, we have
$|\hpot{u}^i - \pot{u}^i| \leq \epsilon\pot{u}^i,$
whenever $\frac{\kappa}{\ell}(1+\epsilon)< \pot{u}^i d_u$. Also, we know $\pmb{x}=\sum_{i=1}^{\ell}\pmb{x}^i$. So, for each $u \in V$ we have $|\hat{\pmb{x}}_u-\pmb{x}_u|=|\sum_{i=1}^{\ell}\hat{\pmb{x}}_u^i-\sum_{i=1}^{\ell}\pmb{x}_u^i|\leq \sum_{i=1}^{\ell}|\hat{\pmb{x}}_u^i-\pmb{x}_u^i|$. Now, let $A_u=\left\{i:\frac{\kappa}{\ell}(1+\epsilon)< \pot{u}^i d_u \right\}$ and $\bar{A_u}=\left\{i:\frac{\kappa}{\ell}(1+\epsilon)\geq \pot{u}^i d_u \right\}$ such that $|A_u|+|\bar{A_u}|=\ell$. We have, $|\hat{\pmb{x}}_u-\pmb{x}_u|\leq \sum_{i\in A_u}\epsilon \pmb{x}_u^i+\sum_{i\in \bar{A_u}}\frac{\kappa}{\ell}(1+\epsilon)\leq \epsilon \pmb{x}_u+\kappa(1+\epsilon)$.
\end{proof}
Note that for the general case we have an additive error of $\kappa(1+\epsilon)$ in addition to the relative error $\epsilon$. The additive error term comes from those node estimates which can't be guaranteed by the \ref{alg:rw_solver} algorithm for the individual one-sink systems. Also, note that for general Laplacian systems with the number of negative entries in $\pmb{b}$ vector as $O(\log n)$, we get the same time complexity guarantees (within polynomial logarithmic factors) as that of the one-sink systems. However, we believe in general  the result presented in Theorem~\ref{thm:rand_walk_solver_gen} can be further improved with a more intelligent way of handling the parallel one-sink systems.

\input{rst_sec}

 \section{Conclusion and Future work}
\label{sec:conc}
Although our main result presents a distributed algorithm, at a deeper level, the key contribution of this paper is not the algorithm we present, but actually the connections our work makes with the queueing theory and ergodicity, and the theory of Markov chains and random walks. Positioning the Laplacian system in a network and solving it there connects the study of Laplacian systems to distributed systems and also raises the possibilities of real-world implementation in low-power networks like sensor networks where such problems are likely to occur naturally. Moving this work ahead, we plan to look at different applications of general Laplacian systems where our proposed solver can be used as subroutine to give improved results. We we would also like to extend our solver for the general setting cases where $\Theta(n)$ entries of $\pmb{b}$ vector are negative.


\bibliographystyle{plain}
\bibliography{main}

\end{document}

%% file: rst_sec.tex
\section{Distributed Generation of Random Spanning Trees}
\label{sec:rst_gen}
In this section, we present an application of our distributed Laplacian solver to generate approximately uniform random spanning trees. We first begin by formally defining our problem of random spanning tree generation and then we present a distributed algorithm to solve it which uses our solver \ref{alg:rw_solver} as a subroutine. Our proposed algorithm is basically a distributed version of Kelner and M\k{a}dry's \cite{KelnerII-FOCS:2009} algorithm. So, we will first give an overview of their approach and also discuss how we adopt it in a distributed setting. Then, we will discuss our algorithm in detail indicating where our distributed solver will be used and then finally we would review the overall complexity of our algorithm.

\subsection{Random Spanning Tree Generation}
\label{subsec:rst_gen_intro}
Given an undirected graph $G=(V,E)$ with $|V|=n$ vertices and $|E|=m$ edges each having unit weight, the random spanning tree generation problem requires us to find an algorithm which outputs a spanning tree $T$ of $G$ with probability $1/T(G)$, where $T(G)$ is the set of all spanning trees of $G$. However, in this paper we look at a relaxed version of this problem, known as $\epsilon$-random spanning tree generation wherein we find an algorithm which generates each spanning tree $T$ with a probability $\prob{T}$ such that $\frac{(1-\epsilon)}{\mathcal{T}(G)}\leq \prob{T}\leq \frac{(1+\epsilon)}{\mathcal{T}(G)}$ i.e., probability of generation is $\epsilon$-away from the uniform distribution. Our proposed algorithm generates these $\epsilon$-random spanning trees in a distributed setting. In particular, our algorithm is a distributed version of Kelner and M\k{a}dry's \cite{KelnerII-FOCS:2009} algorithm based on the famous result of generating uniformly random spanning trees using random walks by Broder \cite{Broder-SFCS:1989} and Aldous \cite{Aldous-SDM:1990}. We give an overview of our approach and compare it with that of Kelner and M\k{a}dry next.

\subsection{Comparative Overview of Kelner and M\k{a}dry's Algorithm and our Approach}
\label{subsec:kelner_algo_dist}
Kelner and M\k{a}dry \cite{KelnerII-FOCS:2009} use random walk-based algorithms for random spanning tree generation. Their approach is based on the result of Broder \cite{Broder-SFCS:1989} and Aldous \cite{Aldous-SDM:1990} who independently showed that if we simulate a random walk on a graph starting from an arbitrary vertex and continue till all vertices are visited, then the set of edges through which each vertex was visited first time by the given walk forms a uniformly random spanning tree of the given graph. However, since this would take time equivalent to the cover time of the graph which can be $O(mn)$ in expectation. Kelner and M\k{a}dry proposed an algorithm to simulate this random walk more efficiently. They observed that the random walk spent a lot of time revisiting the vertices of the graph, so simulating those portions was wasteful. They used the standard ball-growing technique of \cite{Leighton-JACM:1999} to decompose the graph into low diameter partitions which could be quickly covered by the random walk. Then, \emph{for each partition using Laplacian solvers \cite{Spielman-STOC:2004} they precompute the approximate probability of random walk entering that partition from a particular vertex and exiting from a particular vertex as these correspond to the potential developed at the vertex if the given partition is assumed to be an electrical network and the exit vertex is assumed to be attached to a voltage source with a dummy vertex added to the partition assumed to be the sink.} Further, this approximate exit distribution  helps them to shortcut the random walk by removing its trajectories after all vertices of the partition have been visited. Finally, they use these precomputed values to simulate random walk on the graph and generate $\epsilon$-random spanning trees.

Our proposed algorithm \ref{alg:distb_RST_gen} is also based on a similar approach of generating random spanning trees using random walks on the graph. However, we give a  completely distributed algorithm for the problem. We first use a distributed version of Miller et al.'s  \cite{Miller-SPAA:2013} algorithm for decomposition of graph into low diameter partitions $(S_1, \cdots, S_k)$ and set $C$ of edges not entirely contained in one of $S_i$. Miller et al.'s graph decomposition is a parallel version of ball-growing technique with random delays. In this, each node selects a random start time according to some distribution, and if a node is not already part of a partition at that time, it begins its own breadth first search (BFS) to form its cluster. Any node visited by the search which is not part of any partition joins the partition of the node which reached it first and accordingly adds its neighbors to the corresponding BFS queue. Delayed and random start times ensure that the partitions have the desired properties required from the decomposition. We make this algorithm distributed by exchanging messages among the nodes so that all the nodes know the random start times of all other nodes and also use a distributed version of BFS \cite{Awerbuch-SFCS:1985}\cite{Ghaffari-PODC:2014}. Moreover, each node which starts the partition is made the leader of the corresponding partition and is responsible for exchange of synchronization messages among various clusters.  \emph{Once these partitions are made then instead of using Laplacian solver we use our distributed solver to compute the approximate exit distribution for each partition. So, we run \ref{alg:rw_solver} in parallel in each partition $S_i$ to compute the $(1+\epsilon)$-approximation of exit distribution $P_v(e)$ i.e., probability of entering partition through vertex $v$ and exiting from edge $e$ where $v\in V(S_i)$ (vertex set of $S_i$), $e\in C(S_i)$ (set of edges in $C$ with one end-point in $S_i$)}. In addition, we run a random walk on each partition in parallel to compute the spanning tree $\hat{T}_{S_i}$ within each $S_i$ using Broder and Aldous result. The completion of computation step by all $S_i$'s is indicated by exchange of synchronization messages among the leader nodes. After this step, we reduce graph $G$ to $G'$ such that each $S_i$ is assumed to be a super node and we combine all edges $(u,w)$ where $u\in V(S_i)$ and $w\in V(S_j)$ into a super edge connecting two partitions. We then run a random walk on $G'$ with transition probability $\probm{S_i,S_j}=\frac{\sum_{v\in S_i} \sum_{u\in S_i}\sum_{w\in S_j}P_v(u,w)}{\sum_{v\in S_i} \sum_{u\in S_i}\sum_{\forall w}P_v(u,w)}$ where $P_v(u,w)$ are the computed exit distributions representing the probability of a random walk entering a partition $S_i$ through vertex $v\in V(S_i)$ and exiting through edge $e=(u,w)\in C(S_i)$. After that, again using the Broder and Aldous result we obtain spanning tree $\hat{T}_{G'}$ on the reduced graph $G'$. However, note that within each super node random walk takes a predetermined path from the entry vertex $v$ to exit edge $e$ and information about it is exchanged among the nodes in the given partition. Finally, by combining the spanning tree within $S_i$'s i.e., $\hat{T}_{S_i}$ to that of the reduced graph $\hat{T}_{G'}$ we obtain $\epsilon$-random spanning tree of the given graph $\hat{T}_G$.  Refer to Figure~\ref{fig:rand_span_tree} for the explicit changes made to the Kelner and M\k{a}dry's algorithm to adapt it to a distributed setting.

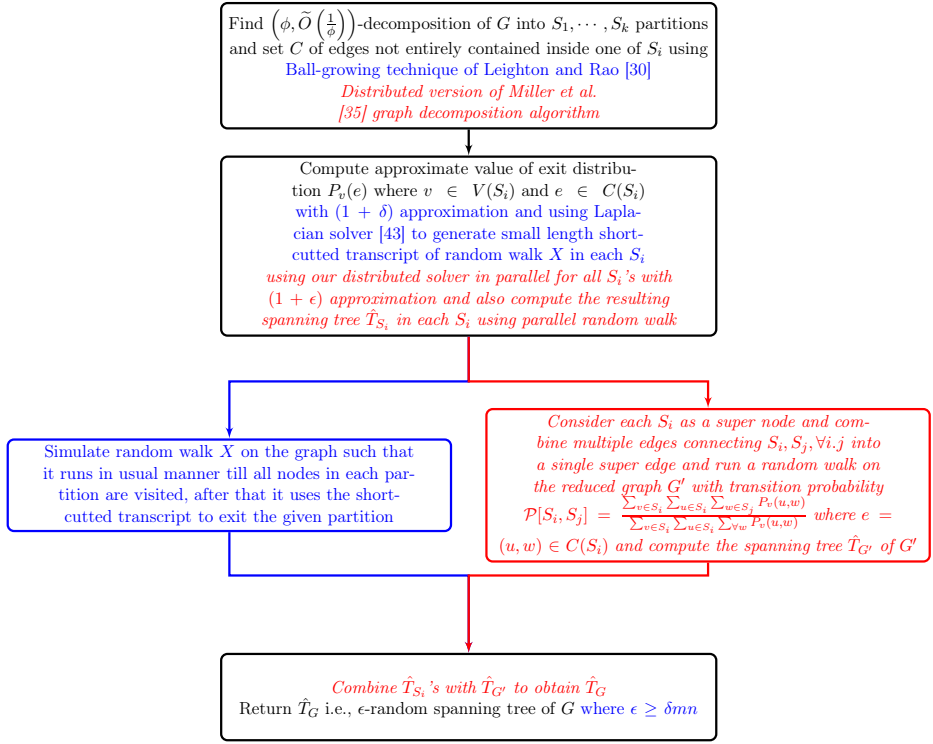
\begin{figure*}
\centering{
\scalebox{0.6}
  {\input{rand_span_tree}}}
  \caption{Comparison of \textcolor{blue}{Kelner and M\k{a}dry's \cite{KelnerII-FOCS:2009} Algorithm} and our \emph{\textcolor{red}{\ref{alg:distb_RST_gen} Algorithm}} for random spanning tree generation (black color denotes common components). }
  \label{fig:rand_span_tree}
\end{figure*}

\subsection{\ref{alg:distb_RST_gen} Algorithm}
\label{subsec:rst_algo_detail}
Having defined our distributed approach for random spanning tree generation and how it differs from Kelner and M\k{a}dry's approach, let us now discuss it in detail. As discussed before, our approach is based on the famous result of Broder and Aldous of generating random spanning trees using random walks. In particular, our proposed algorithm \ref{alg:distb_RST_gen} has three main steps: first of which requires a low-diameter decomposition of given graph into partitions and the last two use Aldous, Broder result with random walks as a basic primitive and our distributed solver as a subroutine to precompute exit distributions for each partition. Let us discuss each of these steps.

\begin{varalgorithm} {Distributed RST Generation}
  \caption{}
  {\label{alg:distb_RST_gen}}
  \begin{algorithmic}[1] 
            \REQUIRE Undirected bounded-degree graph $G=(V,E)$, error parameter $\epsilon$, and $\phi\in(0,1)$
                          \STATE Run \ref{alg:distb_decomp} to obtain $\left(\phi,\widetilde{O}\left(\frac{1}{\phi}\right)\right)$- decomposition of $G$, and set $R$ \\\hfill{\small /* Such that $G$ is decomposed into $S_1, \cdots, S_k$ partitions
                          and set $C$ of edges not entirely contained in one of $S_i$, $R$ is set of leader nodes for all $S_i$ */}
                          \STATE Set vector $\bm{b}$ such that $\bm{b}_1=1$, $\bm{b}_{n_i}=-1$, and $\bm{b}_i=0$ for $i\neq\{1,n_i\}$ \hfill{\small /*where $n_i=|V(S_i)|$ */}
                          \STATE Each $S_i$ runs \ref{alg:rw_solver}$\left(\bm{1},\bm{b}, \epsilon,\frac{1}{\sqrt{n_i}}\right)$ in parallel to compute $(1+\epsilon)$-approximation of exit distribution $P_v(e)$ where $v\in V(S_i)$, $e\in C(S_i)$  \hfill{\small /* where $\bm{1}$ denotes unit weight $\forall e \in E(S_i)$ */}
                          \STATE In parallel, compute the spanning tree $\hat{T}_{S_i}$ within each $S_i$ using a random walk
                          \STATE Leader nodes $i\in R$ of each partition $S_i$ exchange messages to indicate end of their computation step 
                          \IF{Messages from all $k-1$ partitions received by all $i\in R$\hfil{\small /* Messages from everyone other than themselves */} }
                          {
                          \STATE Reduce graph $G$ to $G'$ by considering each $S_i$ as a super node and combine multiple edges connecting two partitions into a single super edge
                          \STATE Run a random walk on the reduced graph $G'$ with transition probability $\probm{S_i,S_j}=\frac{\sum_{v\in S_i} \sum_{u\in S_i}\sum_{w\in S_j}P_v(u,w)}{\sum_{v\in S_i} \sum_{u\in S_i}\sum_{\forall w}P_v(u,w)}$ \hfill{\small /* where $P_v(u,w)$ is the exit distribution for $e=(u,w) \in C(S_i)$ */}
                          \STATE Set of first visited edges by the random walk on $G'$ forms its $\epsilon$-random spanning tree $\hat{T}_{G'}$
                          }
                          \ENDIF
                          \STATE Combine $\hat{T}_{S_i}$ of all $S_i$ with $\hat{T}_{G'}$ to obtain $\hat{T}_G$\hfill
                    {\small /* where $\hat{T}_G$ is $\epsilon$-random spanning tree of $G$ */}
                         \RETURN $\hat{T}_G$
    \end{algorithmic}
 \end{varalgorithm}
\begin{varalgorithm} {Distributed Decomposition $(\phi)$}
  \caption{(Distributed version of Miller et al. \cite{Miller-SPAA:2013})}
  {\label{alg:distb_decomp}}
  \begin{algorithmic}[1] 
            \REQUIRE Undirected bounded-degree graph $G=(V,E)$, parameter $0<\phi<1$
                          \STATE In parallel, each vertex $u$ picks $\delta_u$ independently from an exponential distribution with mean $1/\phi$
                          \STATE In parallel, each node computes  $\delta_{max}=\max\{\delta_u:u \in V\}$ by exchanging $\delta_u$ values
                          \STATE Perform Distributed BFS \cite{Ghaffari-PODC:2014}, with vertex $u$ starting when the vertex at the head of the queue has distance more than $\delta_{max}-\delta_u$
                          \STATE Add vertex $u$ to set $R$ {\small /* $u$ is starting vertex of a partition, $R$ is set of leader nodes for partitions */}
                          \STATE In parallel, each vertex $u$ assigns itself to the point of origin of the shortest path that reached it in the BFS
                          \RETURN $\left(\phi,\widetilde{O}\left(\frac{1}{\phi}\right)\right)$-decomposition of graph and set $R$
    \end{algorithmic}
\end{varalgorithm}
\paragraph{Graph decomposition} The first step of our proposed algorithm is to decompose our given graph into low-diameter partitions which can be easily processed. Let us first formally define the decomposition we require for our algorithm. 
\begin{definition}[$(\phi,\gamma)$-decomposition]
\label{def:low_dia_decomp}
Given a graph $G=(V,E)$ a $(\phi,\gamma)$-decomposition splits it into partitions $(S_1,\cdots,S_k)$ and set $C$ of edges not entirely contained in one of the partitions $S_i$ such that
\begin{itemize}
    \item The diameter of each $S_i$ i.e., $\gamma(S_i)$ is at most $\gamma$, and
    \item $|C| \leq \phi|E(G)|$.
\end{itemize}
\end{definition}

Kelner and M\k{a}dry use ball-growing technique \cite{Leighton-JACM:1999} to obtain $(\phi,\gamma)$-decomposition of graph. However, we will use Miller et al.'s \cite{Miller-SPAA:2013} algorithm in a distributed setting to obtain given decomposition of $G$. 
In Miller et al.'s algorithm, a random shift $\delta_u$ is picked for all nodes from independent exponential distribution with parameter $\phi$. After that, each node is assigned to a partition such that the distance $\delta_{max}-\delta_u$ is minimized where $\delta_{max}=\max_u \delta_u$ . The clusters which will represent our partitions are created using breadth first search (BFS) i.e., if a node $u$ is not already part of a partition by its chosen start time $\delta_u$ then, it starts its own partition by performing a BFS, otherwise the node joins the partition that reached it first. We make this algorithm distributed by exchange of messages among the nodes and using distributed version of BFS \cite{Awerbuch-SFCS:1985}\cite{Ghaffari-PODC:2014}. Moreover, each node which starts the partition is designated to be the leader of that partition and is responsible for exchange of messages on behalf of its partition. The randomized start times chosen by the nodes ensure that the required properties of the $(\phi,\gamma)$-decomposition are satisfied.

First to bound the diameter of partitions, Miller et al. \cite{Miller-SPAA:2013} bound the distance between a node and the leader of the partition to which it is assigned. Since the chosen shift value $\delta_u$ of the leader of partition $S_u$ bounds the distance to any node in $S_u$, so $\delta_{max}=\max_u \delta_u$ is an upper bound on the diameter of each partition. The following lemma gives the bound on the maximum shift value and hence, the diameter of each partition. 
\begin{lemma}[Lemma~4.2, Miller et al. \cite{Miller-SPAA:2013}]
\label{lem:miller_dia}
Given that each node $u\in V$ chooses a random shift value $\delta_u$ from an exponential distribution with parameter $\phi$, the expected value of the maximum shift value $\delta_{max}$ is given by $H_n/\phi$ where $H_n$ is the $n$th harmonic number. Furthermore, with high probability, $\delta_u \leq O\left(\frac{\log n}{\phi}\right)$ for all $u$.
\end{lemma}
For the other property which requires that there are fewer edges between the partitions, Miller et al. show this by bounding the probability that the endpoints of an edge are assigned to different partitions. In particular, Miller et al. prove the following lemma.
\begin{lemma}[Corollary~4.5, Miller et al. \cite{Miller-SPAA:2013}]
\label{lem:miller_edges}
Given a $(\phi,\gamma)$-decomposition of graph $G=(V,E)$ with $|V|=n$ nodes and $|E|=m$ edges into partitions $(S_1,\cdots,S_k)$, the probability of an edge $e = (uv) \in E$ having $u\in V(S_i)$ and $v \in V(S_j)$ such that $S_i\neq S_j$ is bounded by $O(\phi)$, and the expected number of edges between the partitions is $O(\phi m)$.
\end{lemma}
Both these lemmas will hold for our distributed version as well because we only differ from Miller et al. in the way $\delta_{max}$ is computed and BFS is performed by the nodes.

\paragraph{Computation within partitions}
Now, given that we have low diameter partitions of the graph, we will use them along with the following famous result by Broder, Aldous for $\epsilon$-random spanning tree generation.
\begin{lemma}[Broder \cite{Broder-SFCS:1989}, Aldous \cite{Aldous-SDM:1990}]
\label{lem:aldous_broder_rw_rst}
Given an undirected graph $G=(V,E)$ suppose you start a random walk from an arbitrary vertex $u\in V$ and let $T$ be the set of edges used by the walk for the first visit to each vertex $v \in V\setminus\{u\}$, then $T$ forms a uniformly random spanning tree of $G$.
\end{lemma}
We know for our low diameter partitions random walk will take less time to visit all the vertices in it, however, after all the first visits it may still spend a lot of time visiting already covered regions. To avoid these unnecessary steps we need to somehow shortcut the random walk once all nodes in a region are visited. Kelner and M\k{a}dry \cite{KelnerII-FOCS:2009} suggested to compute exit distributions from a given partition $P_v(e)$ where $v\in V(S_i)$ and $e\in C(S_i)$ i.e., the probability of random walk leaving partition $S_i$ through edge $e$ after it had entered it through vertex $v$. Given a $(\phi,\gamma)$-decomposition of graph they compute $(1+\epsilon)$-approximation of all $P_v(e)$ using Laplacian solvers. In particular, they use the following construction for $P_v(e)$ computation: Given a partition $S_i$ and edge $e=(u,u')\in C(S_i)$ with $u\in V(S_i)$ they construct $S_i'$ by adding vertex vertex $u'$ and some dummy vertex $u^*$ to $S_i$. Then, for each boundary edge $(w,w')\in C(S_i)\setminus\{e\}$ with $w\in V(S_i)$, they add an edge $(w,u^*)$ and finally add edge $e=(u,u')$ to it. After the given construction, they treat $S_i'$ as an electric circuit and impose voltage of $1$ at $u'$ and $0$ at $u^*$ and then using Laplacian solvers compute the approximate electrical flow in it  wherein the voltage achieved at any node $v \in V(S_i)$ is equal to $P_v(e)$ (see Lemma~9 \cite{KelnerII-FOCS:2009} for details). We will also use a similar construction, however, to compute these flows and the resulting potentials we will use our distributed solver and we will compute these values in parallel in all partitions. Moreover, at the same time we will run a random walk in each partition in parallel so that from Lemma~\ref{lem:aldous_broder_rw_rst} we obtain uniform spanning tree for each $S_i$. As we will discuss in detail in Section~\ref{subsec:complx_rst}, by setting $\phi=1/\sqrt{n}$ in each partition $S_i$, by the time random walk-based  solver computes $P_v(e)$ for all $e\in C(S_i)$ the parallel random walk would have covered the entire partition, hence, giving us the desired uniformly random spanning tree for $S_i$ along with the exit distributions.

\paragraph{$\epsilon$-random spanning tree generation}
Once all partitions have computed their spanning trees and exit distributions, their respective leaders exchange synchronization messages to indicate the completion of computation step. After messages from all $k$ partitions are received by the leader nodes, we have obtained uniform spanning tree $\hat{T}_{S_i}$ for all $S_i$'s as well as their exit distributions. So, to obtain random spanning tree for $G$ we need to find edges between the different $\hat{T}_{S_i}$'s. For this we reduce graph $G$ to $G'$ such that each partition $S_i$ represents a super node and multiple edges connecting two partitions are combined into a single super edge i.e, we combine all edges $(u,w)$ where $u\in V(S_i)$ and $w\in V(S_j)$. We then use the computed exit distributions of each partition to run a random walk on $G'$ with transition probability $\probm{S_i,S_j}=\frac{\sum_{v\in S_i} \sum_{u\in S_i,w\in S_j}P_v(u,w)}{\sum_{v\in S_i} \sum_{u\in S_i, \forall w}P_v(u,w)}$ where $P_v(u,w)$ is the probability of a random walk entering a partition $S_i$ through vertex $v\in V(S_i)$ and exiting through edge $(u,w)\in C(S_i)$. However, note that within each such super node $S_i$, the random walk will take a predetermined path based on its spanning tree inside $S_i$ and thus would need to exchange messages between nodes and the leader of partition. So, each such step would take $O(\diams(S_i))$ time where $\diams(S_i)$ is the diameter of partition $S_i$.  After the cover time of this random walk, since all vertices of $G'$ will be visited we will obtain random spanning tree $\hat{T}_{G'}$ for $G'$. Finally, combining this random spanning tree with that of all $S_i$'s, we obtain $\epsilon$-random spanning tree of $G$.

\subsection{Overall Complexity}
\label{subsec:complx_rst}
Now, let us review the overall complexity of our algorithm. Our proposed algorithm has three main parts: decomposing graph into low diameter partitions, then using random walk on those partitions to obtain spanning trees within them and finally finding the edges between the partitions to compute the overall spanning tree. For the first part we compute $(\phi,\widetilde{O}(1/\phi))$-decomposition of the graph into $(S_1,\cdots,S_k)$ partitions using distributed version of Miller et al.'s algorithm which takes about $O(m+\diam)$ time where factor $\diam$ (diameter of graph) comes from the distributed BFS and message exchanges and factor $m$ comes from verifying the decomposition. Then, for each partition $S_i$ we compute the exit distribution which corresponds to the node potentials using our proposed distributed solver (Theorem~\ref{thm:rand_walk_solver}) which computes the estimates of node potentials $\hpot{u}, \forall u \in V(S_i)$ such that $|\hpot{u}-\pot{u}| \leq (\epsilon_1+\epsilon_2)\pot{u}, \forall \kappa<\pot{u}d_u$ where $0<\kappa<1$ and $\pot{u}=\bm{\eta}_u/d_u$ in time 
$$\left(d_{\max}^{S_i}\left(64\thit^{S_i}\log \epsilon_1^{-1}+ \frac{4\log n_i}{\kappa^2\epsilon_2^2}\right)\log\frac{16d_{\max}}{3\lambda_2^L}\right)$$
where $d_{\max}^{S_i}$ is the maximum degree of nodes within partition $S_i$, $\thit^{S_i}$ is the worst-case hitting time of random walk on partition $S_i$ with $|V(S_i)|=n_i$ as the size of the partition. Note that $d_{\max}^{S_i}$ term in the running-time of our solver is a constant as we deal with bounded-degree graphs. 
So, the time reduces to $\widetilde{O}\left(\thit^{S_i}\log \epsilon_1^{-1}+\frac{\log n_i}{\kappa^2\epsilon_2^2}\right)$.

Now, let us consider $\phi=1/\sqrt{n}$, so from the first step we obtain a $(1/\sqrt{n},\widetilde{O}(\sqrt{n}))$-decomposition of the graph into partitions $(S_1,\cdots,S_k)$. 
Moreover, within each partition $S_i$ to compute the node potentials accurately we can set the values of $\kappa=1/\sqrt{n_i}$ and given error parameter $\epsilon<1/2$ such that $\epsilon=\epsilon_1+\epsilon_2$, set $\epsilon_1=\epsilon/4$ and $\epsilon_2=3\epsilon/4$. Also, since the worst-case hitting time for any graph is $\Omega(n)$ \cite{Aldous-BOOK:2002} we get the running time of solver for each edge $e\in C(S_i)$ in partition $S_i$ as $\widetilde{O}\left(\thit^{S_i}\log \epsilon^{-1}\right)$ where $\epsilon$ is the error. This is repeated for all edges in set $C(S_i)$ (subset of $C$ incident to $S_i$), so the total time taken for computation of exit distribution for a given partition $S_i$ is $\widetilde{O}\left(|C(S_i)|\left(\thit^{S_i}\log \epsilon^{-1}\right)\right)\leq \widetilde{O}\left(\phi m \left(\thit^{S_i}\log \epsilon^{-1}\right)\right)$, as from $(\phi,\gamma)$-decomposition $|C|\leq \phi m$.
Moreover, the random walk that we run in parallel in this partition to find its spanning tree will take at most cover time $t_{cov}(S_i)$ to visit all vertices. So, for each partition $S_i$, time taken to compute the exit distributions and the spanning tree is $\max\left\{\widetilde{O}\left(\phi m \left(\thit^{S_i}\log \epsilon^{-1}\right)\right), t_{cov}(S_i)\right\}$. 

 From Aleliunas et al. \cite{Aleliunas-SFCS:1979} we know that the cover time of an unweighted graph $G$ is at most $O(|E(G)|\diam)$ where $\diam$ is the diameter of graph, so for $S_i$ we have $t_{cov}(S_i)\leq m_i \sqrt{n}$ as the diameter of each partition is at most $\sqrt{n}$.
 As, the worst-case hitting time of graph is greater than that of its partition i.e., $\thit^{S_i}\leq \thit$, we have $\widetilde{O}\left(\phi m \left(\thit^{S_i}\log \epsilon^{-1}\right)\right)\leq \widetilde{O}\left(\phi m \left(\thit\log \epsilon^{-1}\right)\right)$. So, by the time our random walk-based  solvers compute the exit distributions, the  random walk running in parallel has covered all vertices to give us the random spanning tree of each $S_i$. Now, since we do this computation step in parallel for all partitions, we have the overall time as
 
 $\max_i\left\{\max\left\{\widetilde{O}\left(\phi m \left(\thit\log \epsilon^{-1}\right)\right), t_{cov}(S_i)\right\}\right\} \leq \widetilde{O}\left(\phi m \left(\thit\log \epsilon^{-1}\right)\right)$. Moreover, after each partition completes the computation step its leader exchanges synchronization messages with other leader nodes which takes about $O(k\diam)$ time where $\diam$ is the diameter of the graph and $k$ are the total number of partitions. Once all partitions know that the computation step is over and they proceed to the last step. In the final step, we run the random walk on the reduced graph $G'$ and it takes at most $t_{cov}(G')=O(k^3)$ to form the spanning tree where $k$ is the number of partitions. However, since in each step of this walk within the super nodes we need to communicate entry and exit points of the partition which takes at most $\diam(S_i) \leq \sqrt{n}$ time, so overall time for this step is $O(k^3\sqrt{n})$. Now, we know from the property of our $(\phi,\gamma)$-decomposition that $|C|\leq\phi |E|$ and we have chosen $\phi=1/\sqrt{n}$. Also, $|E|\leq n~d_{\max}/2$ where $d_{\max}$ is the maximum degree of graph $G$. So, to ensure that the graph is connected we have $k \leq \sqrt{n}d_{\max}/2$. Thus, for bounded-degree graphs $k=O(\sqrt{n})$. So, our total time for \ref{alg:distb_RST_gen} is composed of
 \begin{itemize}
     \item Decomposition of graph into low diameter partitions $= O(m+\diam)$.
     \item Using random walks to compute spanning trees within those partitions $=\widetilde{O}\left(\phi m \left(\thit\log \epsilon^{-1}\right)\right)$.
     \item Computing the overall spanning tree
     \begin{itemize}
         \item Exchange of synchronization messages between leaders $=O(k ~\diam)$.
         \item Cover time of random on the reduced graph $=O(k^3\sqrt{n})$.
     \end{itemize}
 \end{itemize}
 As discussed above, since $\phi=1/\sqrt{n}$, $k=O(\sqrt{n})$ for bounded-degree graphs, and $\diam\leq n$, we have the overall time as
 $ O(m)+\widetilde{O}\left(\frac{m}{\sqrt{n}}\thit\log \epsilon^{-1}\right)+O(\diam\sqrt{n})+O(n^2)=\widetilde{O}\left(\frac{m}{\sqrt{n}}\thit\log \epsilon^{-1}\right)$. 

%% file: rand_span_tree.tex
\tikzstyle{block_l} = [rectangle, draw,line width=0.5mm, text width=28em, text centered, rounded corners, minimum height=5em]
\tikzstyle{block_lr} = [rectangle, draw,line width=0.5mm, text width=25em, text centered, rounded corners, minimum height=5em,color=red]
\tikzstyle{block_lb} = [rectangle, draw,line width=0.5mm, text width=25em, text centered, rounded corners, minimum height=5em,color=blue]

\tikzstyle{line} = [draw, -latex']

\begin{tikzpicture}[node distance = 2cm, auto]
     \node [block_l] (decomp) {Find $\left(\phi,\widetilde{O}\left(\frac{1}{\phi}\right)\right)$-decomposition of $G$ into $S_1,\cdots,S_k$ partitions and set $C$ of edges not entirely contained inside one of $S_i$ using\\ \textcolor{blue}{Ball-growing technique of Leighton and Rao \cite{Leighton-JACM:1999}}\\ \emph{\textcolor{red}{Distributed version of Miller et al. \cite{Miller-SPAA:2013} graph decomposition algorithm}}};
     \node [block_l, below of=decomp, node distance=4cm] (shortcut) {Compute approximate value of exit distribution $P_v(e)$ where $v\in V(S_i)$ and $e\in C(S_i)$\\
     \textcolor{blue}{with $(1+\delta)$ approximation and using Laplacian solver \cite{Spielman-SICOMP:2011} to generate small length shortcutted transcript of random walk $X$ in each $S_i$}\\  \emph{\textcolor{red}{using our distributed solver in parallel for all $S_i$'s with $(1+\epsilon)$ approximation and also compute the resulting spanning tree $\hat{T}_{S_i}$ in each $S_i$ using parallel random walk}}};
     \node [block_lb, below left of=shortcut, node distance=7.5cm] (sim){Simulate random walk $X$ on the graph such that it runs in usual manner till all nodes in each partition are visited, after that it uses the shortcutted transcript to exit the given partition };
     \node [block_lr, below right of=shortcut, node distance=7.5cm](rw){\emph{Consider each $S_i$ as a super node and combine multiple edges connecting $S_i,S_j, \forall i.j$ into a single super edge and run a random walk on the reduced graph $G'$ with transition probability $\probm{S_i,S_j}=\frac{\sum_{v\in S_i} \sum_{u\in S_i}\sum_{w\in S_j}P_v(u,w)}{\sum_{v\in S_i} \sum_{u\in S_i}\sum_{\forall w}P_v(u,w)}$ where 
     $e=(u,w) \in C(S_i)$ and compute the spanning tree $\hat{T}_{G'}$ of $G'$}};
     \node [block_l, below of=shortcut, node distance=10cm](rst){\emph{\textcolor{red}{Combine $\hat{T}_{S_i}$'s with $\hat{T}_{G'}$ to obtain $\hat{T}_G$}}\\Return $\hat{T}_G$ i.e., 
     $\epsilon$-random spanning tree of $G$ \textcolor{blue}{where $\epsilon\geq\delta mn$}};
    \path [line width=0.5mm, line] (decomp) -- (shortcut);
    \path [line width=0.5mm,line, color=blue] (shortcut) -- +(0,-3)-| (sim);
    \path [line width=0.5mm,line, color=red] (shortcut) -- +(0,-3)-| (rw);
    \path [line width=0.5mm,line, color=blue] (sim) -- +(0,-2)-| (rst);
    \path [line width=0.5mm,line, color=red] (rw) -- +(0,-2)-| (rst);
\end{tikzpicture}
